\documentclass[11pt]{article}

\usepackage[utf8]{inputenc} 
\usepackage[T1]{fontenc}    
\usepackage{hyperref}       
\usepackage{url}            
\usepackage{booktabs}       
\usepackage{amsfonts}       
\usepackage{nicefrac}       
\usepackage{microtype} 

\usepackage{caption}     

\usepackage{mathrsfs}
\usepackage{multirow}
\usepackage{amstext}
\usepackage{amsmath}
\usepackage{amsthm}
\usepackage{amssymb}
\usepackage{bm}
\usepackage{wrapfig}
\usepackage{graphicx}
\usepackage{bbm, comment}
\usepackage{subfigure}
\usepackage{color}
\usepackage{float}
\usepackage{wrapfig}
\usepackage{enumitem}
\usepackage{url}

\usepackage[ruled]{algorithm2e} 

\makeatletter
\def\BState{\State\hskip-\ALG@thistlm}
\makeatother

\usepackage{scalerel,stackengine}
\stackMath
\newcommand\reallywidehat[1]{%
\savestack{\tmpbox}{\stretchto{%
  \scaleto{%
    \scalerel*[\widthof{\ensuremath{#1}}]{\kern-.6pt\bigwedge\kern-.6pt}%
    {\rule[-\textheight/2]{1ex}{\textheight}}
  }{\textheight}%
}{0.5ex}}%
\stackon[1pt]{#1}{\tmpbox}%
}
\usepackage{tabulary}
\usepackage{booktabs}

\usepackage[top=1.in, bottom=1.in, left=1.in, right=1.in]{geometry}
\usepackage[numbers]{natbib}

\newtheorem*{theorem*}{Theorem}
\newtheorem{theorem}{Theorem}[section]
\newtheorem{lemma}[theorem]{Lemma}
\newtheorem{fact}[theorem]{Fact}
\newtheorem{claim}[theorem]{Claim}
\newtheorem{assumption}[theorem]{Assumption}

\newtheorem{corollary}[theorem]{Corollary}
\newtheorem{example}[theorem]{Example}
\newtheorem{proposition}[theorem]{Proposition}
\newtheorem{definition}[theorem]{Definition}
\newtheorem{observation}[theorem]{Observation}

\newcommand{\dmi}{\textsc{DMI}}
\newcommand{\fmi}{\textsc{FMI}}
\newcommand{\bmi}{\textsc{BMI}}
\newcommand{\mi}{\textsc{MI}}
\newcommand{\vol}{\mathrm{Vol}}
\newcommand{\E}{\mathrm{E}}

\newcommand{\vmi}{\textsc{VMI}}

\newcommand{\mvec}{\mathrm{vec}}

\newcommand{\ube}{\textsc{UBE}}

\begin{document}
\title{More Dominantly Truthful Multi-task Peer Prediction with a Finite Number of Tasks}
\author{Yuqing Kong \\The Center on Frontiers of Computing Studies,\\Peking University \\ \texttt{yuqing.kong@pku.edu.cn} \\}
\date{}
\maketitle
\begin{abstract}
In the setting where we ask participants multiple similar possibly subjective multi-choice questions (e.g. Do you like Bulbasaur? Y/N; do you like Squirtle? Y/N), peer prediction aims to design mechanisms that encourage honest feedback without verification. A series of works have successfully designed multi-task peer prediction mechanisms where reporting truthfully is better than any other strategy (dominantly truthful), while they require an infinite number of tasks. A recent work proposes the first multi-task peer prediction mechanism, Determinant Mutual Information (DMI)-Mechanism, where not only is dominantly truthful but also works for a finite number of tasks (practical). 

However, the existence of other practical dominantly-truthful multi-task peer prediction mechanisms remains to be an open question. This work answers the above question by providing
\begin{itemize}
        \item a new family of information-monotone information measures: volume mutual information (VMI), where DMI is a special case;
      \item a new family of practical dominantly-truthful multi-task peer prediction mechanisms, VMI-Mechanisms.
\end{itemize}

To illustrate the importance of VMI-Mechanisms, we also provide a tractable effort incentive optimization goal. We show that DMI-Mechanism may not be not optimal but we can construct a sequence of VMI-Mechanisms that are approximately optimal.

The main technical highlight in this paper is a novel geometric information measure, \emph{Volume Mutual Information}, that is based on a simple idea: we can measure an object $A$'s information amount by the number of objects that is less informative than $A$. Different densities over the object lead to different information measures. This also gives Determinant Mutual Information a simple geometric interpretation. 
\end{abstract}

\newpage

\section{Introduction}

Human evaluation is a commonly used measure when we lack an objective standard. For example, the internet company sometimes uses human evaluation to evaluate the online product's quality (e.g. app, online platform). However, eliciting high-quality feedback from the human evaluators can be tricky when they are asked to provide \emph{subjective} judgment. There is no way to verify their subjective opinions. Paying these evaluators only for the agreement will discourage valuable feedback from the minority. Peer prediction (i.e. information elicitation without verification) \cite{MRZ05}, aims to design mechanisms that encourage honest subjective feedback from the user, even she is in the minority. In the setting where two users, say Alice and Bob, are asked to rate several similar products (e.g. restaurants), the peer prediction reward system will take their feedbacks as input and return them proper rewards. We want the reward system to be dominantly truthful. That is, for each user (who can belong to a minority group), regardless of other people's behaviors, she will obtain the highest amount of expected reward when she tells the truth and she will be paid the lowest in expectation if she reports some garbage feedback like five stars for all products.

To design dominantly truthful reward systems, \citet{Kong:2019:ITF:3309879.3296670} propose an information-theoretic framework, Mutual Information Paradigm (MIP), to reduce the above mechanism design problem to the design of proper information measure. When the rating tasks are similar, we can assume that Alice and Bob' feedback for these tasks are i.i.d. samples of random variables $\hat{X}_A,\hat{X}_B$. MIP pays Alice and Bob the mutual information between $\hat{X}_A,\hat{X}_B$ in expectation. The mutual information measure should be information-monotone. That is, any data-processing method performed on the random variables will decrease the mutual information. When MIP pays an information-monotone mutual information, the strategic behavior of Alice or Bob will decrease their expected payments since the strategy is a data-processing method. Thus, to design a dominantly truthful mechanism, it is sufficient to design an information measure which 1) is information-monotone; 2) can be estimated unbiasedly with a certain amount of samples.

The original Shannon mutual information satisfies the monotonicity property. However, it cannot be estimated unbiasedly with a finite number of samples thus cannot be used to construct the reward system that works for a finite number of tasks. A recent work \cite{kong2020dominantly} solves this issue by proposing a new mutual information measure, Determinant Mutual Information (DMI). Its corresponded mechanism, DMI-Mechanism, is dominantly truthful with only a finite number of tasks. The trick is that DMI's square has a polynomial format and the polynomial mutual information can be estimated unbiasedly with a finite number of tasks. DMI-Mechanism shows the existence of the finite-number-task dominantly truthful mechanism. However, The existence of other\footnote{Other mechanisms means that these mechanisms are not simple transformations (e.g. affine transformation) of the DMI-Mechanism or based on a mutual information which is a polynomial of DMI (e.g. $\dmi^4+\dmi^2$).} finite-number-task dominantly truthful mechanisms remains to be an open question.

This work answers the above question by providing

\begin{itemize}
        \item a new family of information-monotone information measures: volume mutual information (VMI), where DMI is a special case;
      \item a new family of practical dominantly-truthful multi-task peer prediction mechanisms, VMI-Mechanisms.
\end{itemize}

The family of mechanisms is constructed via the new mutual information family. In detail, to obtain the above results, the paper first formally show that every degree $d$ \emph{polynomial mutual information} can be used to construct the dominantly truthful peer prediction mechanisms that work for $\geq d$ tasks. Most previous information measures are based on \emph{distance method}. The construction of these measures rely on proper distance measures. However, these distance measures based mutual information do not have a polynomial format. This work proposes a novel geometric information measure design framework, \emph{volume method}, to construct a new mutual information family, VMI. Previously, the square DMI is the only known polynomial mutual information even in the binary case. VMI contains a family of new\footnote{A polynomial mutual information is new if it is not a polynomial of DMI (e.g. $\dmi^4+\dmi^2$).} polynomial mutual information. We use these new polynomial mutual information measures to construct the new dominantly truthful peer prediction mechanisms that work for a finite number of tasks. To illustrate this new mutual information family better, we also provide a geometric visualization in the binary case. The visualization provides a deeper understanding of the existed and new mutual information. For example, although the noise decreases the mutual information, the visualization shows that the original Shannon mutual information punishes the two-sided noises more than DMI, and punish the one-sided noises less than DMI. 

Though this work is mainly motivated for answering the above open question, the volume mutual information is the main technical highlight of this work. The idea behind VMI is simple and natural. Given a pair of random variables $X,Y$, mutual information measure takes $X$ and $Y$'s joint distribution as input and outputs their mutual information. Here $(X',Y)$ is less informative than $(X,Y)$ if we can perform an operation on $X$ to obtain $X'$ and this operation is independent of $Y$. A mutual information measure is information-monotone if the mutual information between $X'$ and $Y$ is less than that between $X$ and $Y$. VMI measures how informative a distribution is by measuring the volume of distributions that is less informative than it. That is, the volume mutual information between $X$ and $Y$ is defined as follows: \[\vmi(X;Y):=\text{Volume}(\{(X';Y)|(X';Y)\preceq (X;Y)\}).\] Like other mutual information, volume mutual information operates on $X$ and $Y$'s joint distribution. By assigning different densities to the space of joint distributions, we can obtain different formulas of volume mutual information with different properties. In particular, when the density function is a polynomial of the elements in joint distribution, we can obtain a family of polynomial volume mutual information as well. 

Given a family of practical mechanisms, we have an optimization space. We then provide a tractable optimization goal and optimize over this family. If the participants do not need any effort to perform the tasks, we will focus on incentivizing the participants to tell the truth after they receive the signals. In this case, there is no need to construct more dominantly truthful, practical mechanisms. Thus, we consider the setting where participants require efforts to perform the tasks. In this setting, we want the participants not only to be honest after they have the signals but also to spend a certain amount of effort in obtaining the signals. We assume that the task requester has value for the elicited answers' distribution. We aim to maximize the requester's utility, which is defined as her value minus her payments for the participants. This work's analysis focuses on the setting where there are two participants, Alice and Bob. 

It's left to optimize over the new VMI-Mechanisms. One way is to directly optimize over the new family. Another way is to optimize over all possible dominantly truthful mechanisms first. Then we can approximate the optimal mechanism (may not be practical) by a sequence of practical, dominantly truthful VMI-Mechanisms. It turns out the second way is easier for this problem. First, we observe that the optimal dominantly truthful payment scheme is a threshold payment scheme: there is a threshold joint distribution $U^*$ and if Alice and Bob's reports' joint distribution is more informative than $U^*$, they will get a fixed amount of payments, otherwise, they get nothing. This payment scheme only works for an infinite number of tasks where we can perfectly estimate Alice and Bob's reports' joint distribution. However, there exists a sequence of practical VMI-Mechanisms that approximate the optimal threshold payment scheme. The idea is that the threshold payment scheme is a special VMI-Mechanism if we allow the density function to be a Dirac delta function on $U^*$. To construct a sequence of practical VMI-Mechanisms to approximate the threshold payment scheme, we use a sequence of polynomials to approximate the Dirac delta function. In the literature of proper scoring rules, there is a beta family of scoring rules \cite{2005Loss,2013Choosing} which can be used to approximate a threshold scoring rule, ``misclassification'' scoring. We are inspired to pick the multivariate Beta (Dirichlet) family to design a parametric family of VMI and use this family to approximate the optimal threshold payment scheme.

\begin{figure}[!ht]
\begin{center}
\includegraphics[height = 5cm]{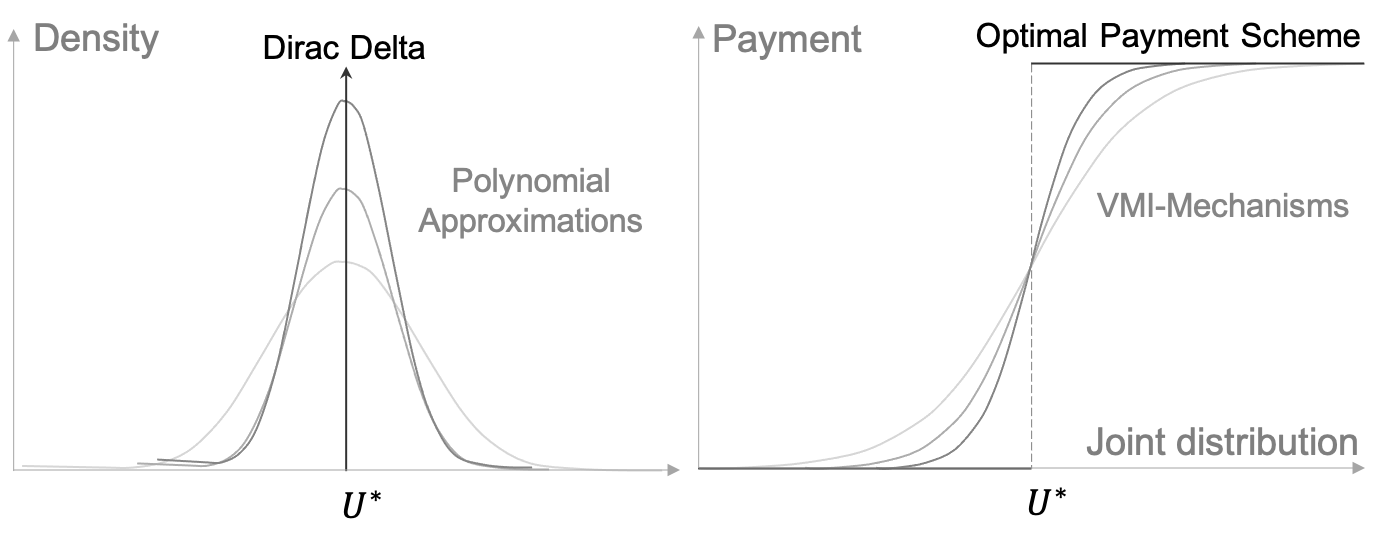}
\end{center}
\captionsetup{singlelinecheck=off} 

\caption[]{\textbf{An illustration of optimizing multi-task peer prediction mechanism}: the above figure illustrates the conceptual idea of optimizing the multi-task peer prediction mechanism. For ease of illustration, we draw the space of the joint distributions as a line though, in fact, it is not. The optimal payment scheme is a threshold function. A VMI-Mechanism corresponds to a density function over the space of the joint distributions. The Dirac delta density leads to the threshold payment scheme. We use a sequence of polynomial densities to approximate the Dirac delta density and use those densities to construct corresponding VMI-Mechanisms. Then we obtain a sequence of practical and dominantly truthful VMI-Mechanisms that approximate the optimal payment scheme.} \label{fig:illustraction}
\end{figure}

Thus, we not only contribute a new family of practical dominantly truthful multi-task peer prediction mechanisms, but also provide a tractable effort incentive optimization goal. We show that under this goal, DMI-Mechanism may not be optimal but we can use our new family to construct a sequence of approximately optimal practical dominantly truthful mechanisms.

\subsection{Related Work}

\citet{MRZ05} start the literature of peer prediction by considering the setting where the participants are asked a single question (e.g. do you like this restaurant or not?). They design a reward system where truth-telling is a strict equilibrium. However, this original peer prediction work requires the knowledge of the common prior over the participants. \citet{prelec2004bayesian} proposes the Bayesian Truth Serum that removes this prior knowledge requirement by asking the participants to additionally report their forecasts for other people (e.g. what percentage of your peers like this restaurant?). However, Bayesian Truth Serum requires an infinite number of participants. Moreover, the additional forecast report requires additional efforts from the participants and makes the mechanism \emph{non-minimal}. \citet{radanovic2015incentive} extend Bayesian Truth Serum to a slightly different setting involving sensors, but still requires a large number of agents. A series of works (e.g. \cite{radanovic2014incentives,faltings2014incentives,witkowski2012robust,DBLP:conf/innovations/KongS18}) study how to remove the requirement for a large number of participants, while their mechanisms are non-minimal. \citet{geometricpp} focus on the design of minimal mechanisms where the participants do not need to report additional forecasts. However, when participants are only assigned a single task, they point out that minimal mechanisms require constraints on the participants' belief model, i.e., are not \emph{prior-independent}. 

\citet{dasgupta2013crowdsourced} start to consider the setting where the participants are assigned multiple similar tasks, the multi-task setting. In contrast to the single-task setting, the multi-task setting enables the design of both \emph{prior-independent} and \emph{minimal} mechanisms. In the multi-task setting, \citet{rfj2016} use the distribution of reported answers from similar tasks as the prior probability of possible answers, while their mechanism requires the estimation of prior probability from a large number of tasks or participants. \citet{kamble2015truth} propose a mechanism where the participants can perform only a single task though the total number of tasks is large. However, this mechanism is not dominantly truthful and makes truth-telling only better than any symmetric equilibrium where all participants perform the same strategy. \citet{2016arXiv160303151S,Kong:2019:ITF:3309879.3296670,liuchen} focus on the setting where there are a small number of participants and show that the dominantly truthful multi-task peer prediction mechanism exists. \citet{Kong:2019:ITF:3309879.3296670} also provide a general information-theoretic framework for the design of the dominantly truthful peer prediction mechanisms. However, they all require the participants to perform an infinite number of tasks. 

\citet{kong2020dominantly} addresses this issue by proposing the first dominantly truthful mechanism, DMI-Mechanism, which is prior-independent, minimal, and works for two participants and a finite number of tasks (practical). This mechanism is constructed by a new information measure, Determinant Mutual Information (DMI) whose square has a polynomial format. However, the existence of other practical dominantly truthful mechanisms remains to be an open question. This work answers the above question by providing a family of practical dominantly truthful peer prediction mechanisms, as well as a new family of information-monotone mutual information: volume mutual information (VMI), where DMI is a special case. 

Regarding optimization in information elicitation, \citet{NNW-20,DBLP:journals/corr/abs-2007-02905,Z-11,2013Choosing, Osb-89} focus on optimizing over proper scoring rules. Unlike this work, in the setting of scoring rules, the ground truth will be revealed later and the participants report only once. \citet{2014Optimum} consider the setting where workers are asked to report a data point and aim to find the optimal statistical estimator with the best effort incentives. We consider a very different setting. \citet{geometricpp} optimize over single-task peer prediction mechanisms where their mechanisms are not dominantly truthful. In contrast, we focus on the multi-task peer prediction setting where ground truth does not exist and the participants will perform multiple tasks. Moreover, we optimize over dominantly truthful, and practical mechanisms. 



\subsection{Multi-task Peer Prediction and Mutual Information}\label{sec:mip}

In this section, we will show how to employ information-monotone mutual information measures to design dominantly truthful mechanisms. We will also connect polynomial mutual information to the practical mechanism. Then we can reduce the design of the dominantly truthful and practical mechanism to information-monotone polynomial mutual information. 

\paragraph{Multi-task Peer Prediction} We focus on the setting where there are two participants, Alice and Bob, and a task requester. Alice and Bob are assigned $T$ a priori similar tasks. For each task $t$, after performing the task, each participant $i=A,B$ will receive a private signal $c_i^t\in \mathcal{C}$ where $\mathcal{C}$ is a size $C$ choice set. For binary questions, $C=2$. By assuming the tasks are a priori similar, the participants' honest signals $\{(c_A^t,c_B^t)\}_t$ are $T$ i.i.d. samples from random variables $(X_A,X_B)$ whose distribution is denoted by $U_{A,B}$. $U_{A,B}$ can be seen as a $C\times C$ matrix where $U_{A,B}(c_A,c_B)$ is the probability that $(X_A,X_B)=(c_A,c_B)$. A multi-task peer prediction mechanism will take all participants' reports $\{(c_A^t,c_B^t)\}_{t=1}^T$ as input and output their corresponding payments $p_A,p_B$. 

\paragraph{Report Strategy Model} Alice may lie and her \emph{strategy} $S_A^t$ for each task $t$ can be seen as a $C\times C$ stochastic matrix where $S_A^t(\hat{c}_A^t,c_A^t)$ is the probability she reports $\hat{c}_A^t$ given that she receives $c_A^t$. We follow \citet{kong2020dominantly} and assume that every participant plays the consistent strategy for all tasks. That is, there exists $S_A$ such that $\forall t, S_A^t=S_A$. We model Bob analogously. With this assumption, not only the participants' honest signals are i.i.d. samples, but also their reported signals are i.i.d. samples from random variables $(\hat{X}_A,\hat{X}_B)$ whose distribution is denoted by $\hat{U}_{A,B}$. A strategy $S$ is uninformative if it is independent of private signals, i.e., $S(\hat{c},c)=S(\hat{c},c')$ for all $c,c',\hat{c}\in \mathcal{C}$. 

\begin{definition}[Dominantly truthful]
A multi-task peer prediction mechanism is dominantly truthful if, for all participants, truthful report strategy maximizes her expected payment regardless of other people's strategies; and if she believes other participants tell the truth, the truthful report strategy will be strictly better than uninformative report strategies.
\end{definition}

The second requirement guarantees that the flat payment mechanism is not dominantly truthful. With the above report strategy model, for a dominantly truthful mechanism where Alice's expected payment is represented as $\mathcal{P}_A(\hat{U}_{A,B})$ and Bob's expected payment is represented as $\mathcal{P}_B(\hat{U}_{A,B})$, we have $\forall S_A, S_B, U_{A,B}$, \[\mathcal{P}_A(S_A U_{A,B} S_B^{\top})\leq \mathcal{P}_A(U_{A,B}S_B^{\top})\] and analogously \[\mathcal{P}_B(S_A U_{A,B} S_B^{\top})\leq \mathcal{P}_B(S_A U_{A,B}).\] 

\citet{Kong:2019:ITF:3309879.3296670} introduce an information-theoretic framework, Mutual Information Paradigm (MIP), for the design of dominantly truthful multi-task peer prediction mechanisms. MIP pays each participant the mutual information between her report and her peer's report. Once the mutual information is information-monotone, each participant will be incentivized to tell the truth to avoid the loss of information.  We start to formally define information-monotonicity. 

Let $U_{X,Y}$ be a joint distribution over two random variables $X$ and $Y$. We want to design an information measure $\mi$ that takes $U_{X,Y}$ as input and outputs a non-negative real number, which reflects the amount of information contained in $X$ that is related to $Y$. We also want $\mi$ to be \emph{information-monotone}: when $X'$ is ``less informative'' than $X$ with respect to $Y$, $\mi(U_{X',Y})$ should be less than $\mi(U_{X,Y})$. Typically, the literature also writes $\mi(U_{X,Y})$ as $\mi(X;Y)$. The following definition is the formal definition of information-monotonicity.

\begin{definition}[Information-monotonicity]\cite{cover2006elements}
    $\mi$ is information-monotone if for every two random variables $X,Y$, when $X'$ is less informative than $X$ with respect to $Y$, i.e., $X'$ is independent of $Y$ conditioning $X$, \[\mi(X';Y)\leq \mi(X;Y).\]
\end{definition}

Mutual information requires the distribution as input while we only have samples. However, since the participants are assumed to be the expected payment maximizer, the unbiased estimator is sufficient.

\paragraph{Unbiased estimator of mutual information} Given a mutual information $\mi$, $\ube^{\mi}$ is an unbiased estimator of $\mi$ with $\geq r$ sample if for every two random variables $(X,Y)$, when $\{(x_t,y_t)\}_{t=1}^T$ are $T\geq r$ independent samples of $(X,Y)$, \[\E[\ube^{\mi}(\{(x_t,y_t)\}_{t=1}^T)]=\mi(X;Y).\]

\paragraph{Mutual Information Paradigm$(\ube^{\mi})$} Alice and Bob are assigned $T\geq r$ a priori similar tasks in independent random orders. The participants finish the tasks without any communication. 
\begin{description}
\item[Report] For each task $t$, Alice privately receives $c_A^t$ and reports $\hat{c}_A^t$ and Bob is analogous. 
\item[Payment] Alice's payment is 
\[ p_A:=  \ube^{\mi}(\{(\hat{c}_A^t,\hat{c}_B^t)\}_{t=1}^{T}) \] where $\ube^{\mi}$ is an unbiased estimator of an information-monotone $\mi$ that works for $\geq r$ samples. Bob is analogous. 

\end{description}

We say agents' prior is informative for $\mi$ if the mutual information tween their truthful reports are positive, i.e., $\mi(X_A;X_B)>0$. This assumption is required to guarantee the second property of dominant truthfulness. 

\begin{lemma}
\label{lem:paradigm}
    When $\mi$ is information-monotone, non-negative, and vanishes for independent random variables, if agents' prior is informative with respect to $\mi$, then the mutual information paradigm $\ube^{\mi}$ is dominantly truthful.
\end{lemma}
\begin{proof}
In expectation, Alice's payment is $ \mi(\hat{X}_A;\hat{X}_B) $ which will be maximized if she tells the truth. If agents' prior is informative with respect to $\mi$ and Alice believes Bob tells the truth, Alice's expected payment when she tells the truth will be $\geq \mi(X_A;X_B)>0$. If she reports uninformative signals, her expected payment will be zero since $\mi$ vanishes for independent random variables. Thus, the second property of dominant truthfulness is also satisfied. 
\end{proof}

To design a practical dominantly truthful mechanism, the unbiased estimator needs to work for only a finite number of samples. We will show that once the mutual information is a degree $d$ polynomial, it has an unbiased estimator that works for $\geq d$ samples. Currently, the only example of polynomial mutual information is $\dmi$'s square. 

\begin{definition}[Polynomial Mutual Information]
    $\mi$ is a polynomial mutual information when $\mi(X;Y)$ a multivariate polynomial of the entries of $U_{X,Y}$.\end{definition}

\begin{definition}[Determinant based Mutual Information (DMI) \cite{kong2020dominantly}]\[\dmi(X;Y):=|\det(U_{X,Y})|\]\end{definition}

DMI is not a polynomial mutual information while DMI's square is. For example, in the binary case for every joint distribution matrix $U_{X,Y}=\begin{bmatrix}u_{00} & u_{01}\\ u_{10} & u_{11}\end{bmatrix}$, $\dmi(X;Y)=|u_{00}u_{11}-u_{10}u_{01}|$ is not a polynomial while $\dmi^2(X;Y)=(u_{00}u_{11}-u_{10}u_{01})^2$ is a polynomial.


\begin{lemma}\label{lem:poly}
    Every degree $d$ polynomial mutual information $\mi$ has an unbiased estimator $\ube^{\mi}$ for $T\geq d$ samples. 
\end{lemma}

\begin{proof}
    Every degree $d$ polynomial mutual information $\mi$ can be written as the sum of terms of format $\Pr[X=c_1,Y=c'_1]*\Pr[X=c_2,Y=c'_2]*\cdots*\Pr[X=c_k,Y=c'_k],k\leq d$.
    
    For each term $\Pr[X=c_1,Y=c'_1]*\Pr[X=c_2,Y=c'_2]*\cdots*\Pr[X=c_k,Y=c'_k],k\leq d$, when we have $k$ independent samples $(x_1,y_1),(x_2,y_2),\cdots,(x_k,y_k)$ of $X,Y$, $\Pi_{i=1}^{k}\mathbbm{1}(x_i=c_i,y_i=c'_i)$ is an unbiased estimator. Thus, since $k\leq d$, $T\geq d$ independent samples is sufficient to construct an unbiased estimator of each term as well as the sum of these terms $\mi$. 
\end{proof}


The above lemma shows that every degree $d$ polynomial mutual information corresponds to a dominantly truthful peer prediction mechanism that works for $\geq d$ tasks. For example, DMI's square is a degree $2C$ polynomial. DMI-Mechanism \cite{kong2020dominantly} is constructed via an unbiased estimator of DMI's square and requires $\geq 2C$ tasks. 


\section{Volume Mutual Information}\label{sec:vmi}

This section will introduce the volume method and apply the volume method to obtain a new family of information-monotone mutual information measure, Volume Mutual Information (VMI), which can be polynomials. 
\subsection{Volume Method}\label{sec:volume}


Given a partially ordered set (poset) $(L,\preceq)$, we define the \emph{lower set} of $\ell$'s as \[\downarrow \ell:=\{\ell'|\ell'\in L, \ell'\preceq \ell \}.\] In discrete case, volume method measures each element by the size of its lower set. In continuous case, we need a monotone measure $\mu$ and integral $\int d\mu$ on $L$. That is, $\mu$ assigns higher volume to bigger set and for two integrable real-valued functions $f\leq g$ on $X$, $\int_X f d\mu\leq \int_X g d\mu$. We defer the basic definitions for measure and integral to appendix.

We assume that all lower sets are measurable with $\mu$. Since the higher-order element has a larger lower set, the volume of each element's lower set \[V(\ell):=\vol(\downarrow \ell):=\mu(\downarrow \ell)\] is a natural monotone function with respect to the partial order. More generally, we define a weighted version:

\begin{definition}[Volume function]
    Given a poset $(L,\preceq)$ with a monotone measure $\mu$ and a monotone integral $\int d\mu$ on $L$, for every integrable non-negative density function $w:L\mapsto \mathbb{R}^+$, we define the volume function that is associated with $w$ as 
\[ V^w(\ell):=\vol^w(\downarrow\ell):=\int_{\downarrow \ell} w(x) d\mu(x) .\] 
\end{definition}

When $w(x)=1$ everywhere, $V^w(\ell)=V(\ell)$.  

\begin{lemma}\label{lem:key}
The volume function $V^w: L\mapsto \mathbb{R}^+$ is a non-negative monotone function. 
\end{lemma}

The above lemma shows that $V^w$ extends a partial order to a total order. 

\begin{proof}

When $\ell'\preceq \ell$, since $\preceq$ is transitive, 
\[\downarrow \ell'\subset \downarrow \ell.\] Due to the fact that the measure and the integral are monotone, $V^w(\ell)$ is also monotone.

\end{proof}

\subsection{Information-monotone Partial Order}\label{sec:monotone}

To apply the volume method to the design of mutual information, we first use information-monotonicity to define a partial order among the joint distributions. $U_{X',Y}\preceq U_{X,Y}$ iff $X'$ is less informative than $X$ with respect to $Y$, i.e, $X'$ is independent of $Y$ conditioning $X$. We will show that this is equivalent to the following definition. 

\begin{definition}[$(L,\preceq)$ for $\mi$]\label{def:lpo}
    We define domain $L$ as the set of all possible $C\times C$ joint distribution matrices. $U'\preceq U$ if there exists a column-stochastic\footnote{A matrix $T$ is a column-stochastic matrix iff every entry of $T$ is non-negative and every column of $T$ sums to 1. } matrix $T$ such that $U'= T U$. 
\end{definition}

\begin{example}

 \[ \begin{bmatrix} .5 & .5\\ .5 & .5\end{bmatrix}U\cong\footnote{If $A\preceq B$ and $B\preceq A$, then $A\cong B$. }\begin{bmatrix} 0 & 0\\ 1 & 1\end{bmatrix}U\preceq \begin{bmatrix} .5 & 0\\ .5 & 1\end{bmatrix}U\preceq U.\]


The first equality holds since $\begin{bmatrix} 0 & 0\\ 1 & 1\end{bmatrix}=\begin{bmatrix} 0 & 0\\ 1 & 1\end{bmatrix} \begin{bmatrix} .5 & .5\\ .5 & .5\end{bmatrix}$ and $\begin{bmatrix} .5 & .5\\ .5 & .5\end{bmatrix}=\begin{bmatrix} .5 & .5\\ .5 & .5\end{bmatrix}\begin{bmatrix} 0 & 0\\ 1 & 1\end{bmatrix} $. The second partial order is valid since $\begin{bmatrix} 0 & 0\\ 1 & 1\end{bmatrix}=\begin{bmatrix} 0 & 0\\ 1 & 1\end{bmatrix} \begin{bmatrix} .5 & 0\\ .5 & 1\end{bmatrix}$. 

\end{example}


The following lemma shows that designing information-monotone mutual information is equivalent to designing a monotone function on $(L,\preceq)$.

\begin{lemma}\label{lem:mono}
$\mi$ is information-monotone if and only if $\mi$ is a monotone function on $(L,\preceq)$. 
\end{lemma}

\begin{proof}
    We first show the $\Leftarrow$ direction. when $X'$ is less informative than $X$ with respect to $Y$, i.e., $X'$ is independent of $Y$ conditioning $X$, \[U_{X',Y}(x',y)=\Pr[X'=x',Y=y]=\sum_x \Pr[X'=x'|X=x]\Pr[X=x,Y=y].\] Thus, $U_{X',Y}=U_{X'|X}U_{X,Y}$. Since $U_{X'|X}$ is a column-stochastic matrix, $U_{X',Y}\preceq U_{X,Y}$. When $\mi$ is a monotone function on $(L,\preceq)$, $\mi$ is information-monotone. 
    
    To show the opposite direction, we start from the situation that $\mi$ is information-monotone. For every $U$, for every column-stochastic matrix $T$, we only need to show there exists $X,X',Y$ such that $X'$ is less informative than $X$ and $U_{X,Y}=U$ and $U_{X',Y}=T U$. We can construct such $X,X',Y$ by setting $\Pr[X=x,X'=x',Y=y]=U_{X,Y}(x,y)U_{X'|X}(x',x)$ for every $x,x',y$. Here $U_{X,Y}(x,y):=U(x,y)$ and $U_{X'|X}(x',x):=T(x',x)$. It's easy to see that $U_{X',Y}=T U$ and $X'$ is less informative than $X$ for $Y$. 
    
    Thus, $\mi(TU)=\mi(X';Y)\leq \mi(X;Y)=\mi(U)$. The inequality follows from the fact that $\mi $ is information-monotone. Therefore, $\mi$ is also monotone on the poset and the $\Rightarrow$ direction is also valid. 
\end{proof}

\subsection{Constructing Volume Mutual Information}\label{sec:defvmi}
This section will apply the volume method to obtain a new family of monotone mutual information measures, Volume Mutual Information (VMI). We have already defined the poset. Thus, to apply the volume method, we only need to pick the measure and integral.  

We will use Hausdorff measure \cite{simon1983lectures}. Intuitively, to provide a measure for any triangle's area on $\mathbb{R}^2$, the 2-dimensional Lebesgue measure $\mathscr{L}^2$ works. However, $\mathscr{L}^2$ will assign zero measure to any curve in $\mathbb{R}^2$. To provide a measure for a curve's length in $\mathbb{R}^2$, we can use the Hausdorff measure $\mathscr{H}^1$. We defer more introduction about the basic measure theory to the appendix. 

\begin{definition}[$(L,\preceq,\mu,\int)$ for $\mi$]
    We define domain $L$ as the set of all possible $C\times C$ joint distribution matrices. $U'\preceq U$ if there exists a column-stochastic matrix $T$ such that $U'= T U$. We vectorize matrices and transform $L$ to space in $\mathbb{R}^{C^2}$. We pick $\mu$ as the $C(C-1)$-dimensional\footnote{Though $L$ is a subset of a $C^2$ dimensional space, in the later sections, we will see the lower set has at most $C(C-1)$ dimension.} Hausdorff measure $\mathscr{H}^{C(C-1)}$. 
    
\end{definition}


\begin{example}[$(L,\preceq,\mu,\int)$ in binary case]

The following observation allows us to visually illustrate $(L,\preceq,\mu,\int)$ for binary case in Figure~\ref{fig:binary}. 

    \begin{observation}\label{obs:binary}
In binary case, there is an one to one mapping from $[0,1]^3$ to $L$. In fact,
\[L=\{\begin{bmatrix} s & t\\ 1-s & 1-t\end{bmatrix} \begin{bmatrix}p & 0\\ 0 & 1-p\end{bmatrix}|s,t,p\in[0,1] \}.\]

Fixing $p$, there is an one to one mapping from $[0,1]^2$ to $\downarrow  U_p$ where $U_p= \begin{bmatrix}p & 0\\ 0 & 1-p\end{bmatrix}$. $\downarrow  U_p$ is the space of all joint distribution matrices whose column sum is $(p,1-p)$ and $L=\{\downarrow  U_p|p\in[0,1]\}$.   

\end{observation}

The proof is deferred to the appendix.

\end{example}

\begin{figure}[!ht]
\begin{center}
\includegraphics[height = 7cm]{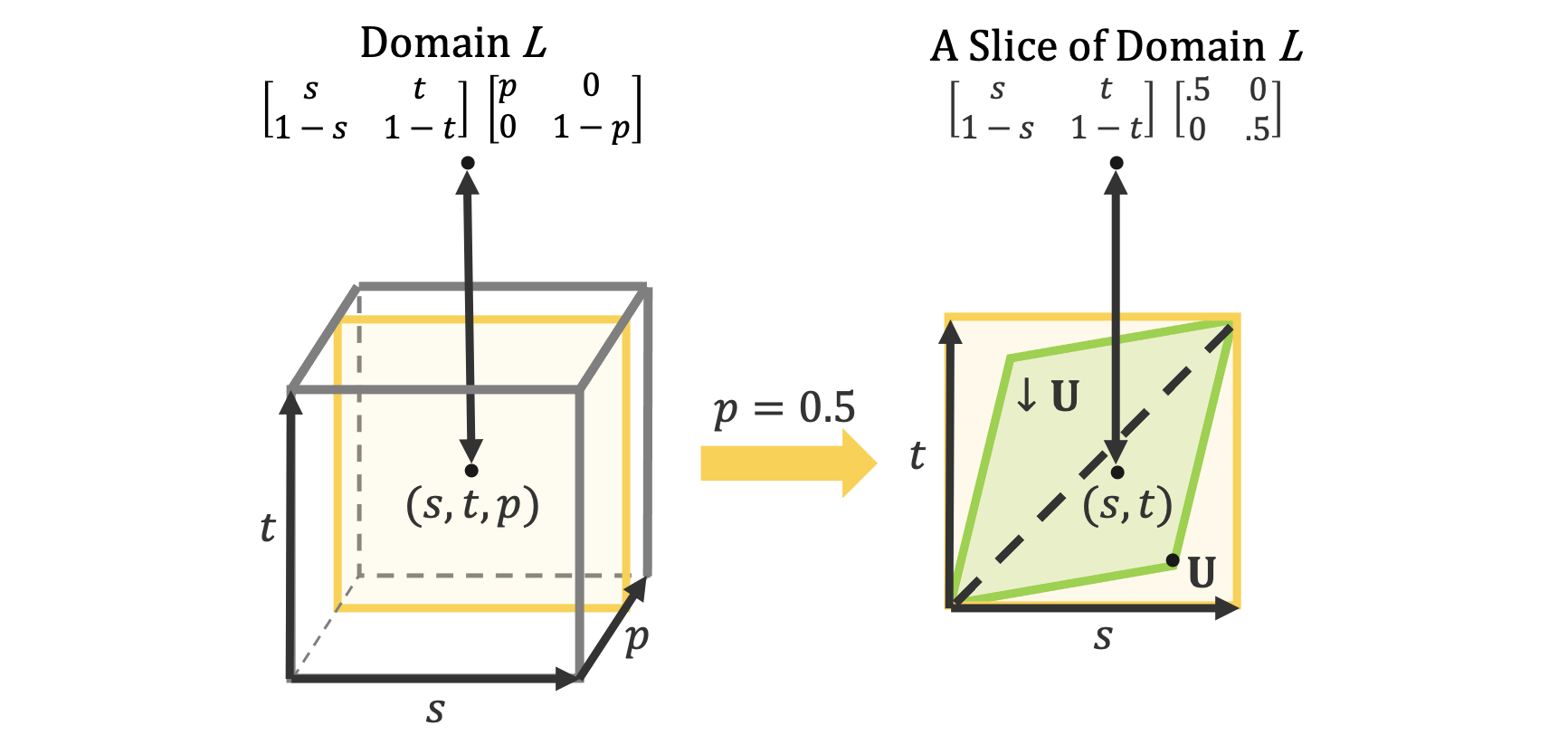}
\end{center}
\captionsetup{singlelinecheck=off} 

\caption[foo bar]{\textbf{Visual illustration $(L,\preceq,\mu,\int)$ in binary case}:
  \begin{itemize}
    \item Domain $L$: there exists a one to one mapping from the domain $L$ to a unit cube $[0,1]^3$. Thus, we visualize $L$ as a unit cube. The right square represents a slice of $L$, $\downarrow  U_{.7}$, the space of all joint distribution matrices whose column sum is $(.7,.3)$.
    \item Lower set $\downarrow U$: for each element $U$, all $U'\preceq U$ constitute a parallelogram (the light green area) whose endpoints are $\{U, \begin{bmatrix} 0 & 1\\ 1 & 0\end{bmatrix}U, \begin{bmatrix} 1 & 1\\ 0 & 0\end{bmatrix}U, \begin{bmatrix} 0 & 0\\ 1 & 1\end{bmatrix}U\}$. This parallelogram is also called $U$'s lower set. 
    \item Uninformative set: when $s=t$ (the black dashed line), the distribution represents independent $X$ and $Y$. In this case, the mutual information should be zero. We call the set of these independent distributions the uninformative set. 
    \item Measure $\mu$: since the lower set is always on a 2-dimensional space, we use the 2-dimensional Hausdorff measure $\mathscr{H}^2$ to measure the area of the parallelogram in $\mathbb{R}^3$. 
\end{itemize} } \label{fig:binary}
\end{figure}


\begin{definition}[Volume Mutual Information $\vmi^{w}$]
    Given an integrable non-negative density function $w$, we define the Volume Mutual Information as \[\vmi^{w}(X;Y):= V^w(U_{X,Y})=\vol^w(\downarrow  U_{X,Y})=\int_{\downarrow  U_{X,Y}}w(x)d\mathscr{H}^{C(C-1)}(x).\] \end{definition}
    
Aided by programming, we can obtain the explicit formula of VMI (Example~\ref{ex:binary}). The choice of density functions affects the property of VMI. Theoretically, we will show that uniform density leads to DMI and polynomial density obtains polynomial VMI (Theorem~\ref{thm:vmi}), which leads to a family of practical dominantly truthful peer prediction mechanisms (Corollary~\ref{coro:main}). Numerically, we will show the influence of density visually by three concrete binary VMI (Example~\ref{ex:binary}). To state the theorem formally, we first give a formal definition for polynomial mutual information.




\begin{theorem}\label{thm:vmi}
$\vmi^{w}$ is an information-monotone mutual information. $\vmi^{w}$ is also non-negative and when $X$ and $Y$ are independent, $\vmi^{w}(X;Y)=0$. Moreover, 
\begin{description}
\item [Uniform density] with the uniform density, $\vmi(X;Y)\propto \dmi(X;Y)^{C-1}$;
\item [Polynomial density] when the density function $w$ is a non-negative degree $d_w$ polynomial, when $C$ is an odd number, $\vmi^D$ is a degree $d_w+C(C-1)$ information-monotone polynomial mutual information and when $C$ is an even number, $\dmi* \vmi^D$ is a degree $d_w+C^2$ information-monotone polynomial mutual information. $(\vmi^D)^2$ is a degree $2(d_w+C(C-1))$ information-monotone polynomial mutual information.
\end{description}
\end{theorem}
Every degree $d$ polynomial monotone mutual information directly induce a dominantly truthful multi-task peer prediction mechanism that works for $\geq d$ tasks (Lemma~\ref{lem:poly}). 

\begin{corollary} \label{coro:main}   There exists a family of practical, dominantly truthful and prior-independent multi-task peer prediction mechanisms. \end{corollary} 

\begin{proof} [Proof of Corollary~\ref{coro:main}] Theorem~\ref{thm:vmi} shows the existence of a family of polynomial mutual information. Lemma~\ref{lem:poly} shows that each degree $d$ polynomial mutual information $\mi$ has an unbiased estimator with $\geq d$ samples. Lemma~\ref{lem:paradigm} shows that when agents' prior is informative for $\mi$, we can use the above unbiased estimator to construct a dominantly truthful peer prediction mechanism that works for $\geq d$ tasks.\end{proof}
We have proved that polynomial VMI can be used to construct practical mechanisms. In Appendix~\ref{sec:newmec}, we will also provide a concrete example for VMI-Mechanism in the binary case.

\paragraph{Proof outline for Theorem~\ref{thm:vmi}} The fact that $\vmi^w$ is information-monotone follows directly from Lemma~\ref{lem:key} and Lemma~\ref{lem:mono}. We will apply the area formula (Fact~\ref{fact:area}) to prove the other parts. With the uniform density, to show that $\vmi(X;Y)\propto \dmi(X;Y)^{C-1}$, we only need to show the original volume of the lower set is proportional to $\dmi(X;Y)^{C-1}$. We will construct a proper affine mapping from $\mathbb{R}^{C(C-1)}$ to $L$ and directly apply the area formula to show this result. To show the last part of this theorem, we will write down the integration explicitly and then analyze it. We defer the full proof to the appendix.

\subsection{Visualization of Binary Volume Mutual Information}\label{sec:binary}

This section will provide a visualization method for all binary mutual information. By using this visualization method, we visualize three new VMIs for three styles of densities (mountain, plain, basin). We additionally visualize two existed mutual information measures in Appendix~\ref{sec:exist}.

\begin{definition}[Contour plots of binary MIs]
    In binary case, the mutual information can be seen as a function with 3 variables:\[\mi^{3d}(s,t,p;\mi):=\mi(\begin{bmatrix} s & t\\ 1-s & 1-t\end{bmatrix}  \begin{bmatrix} p & 0\\ 0 & 1-p\end{bmatrix}).\] To visualize the contour plot in a 2 dimensional space, we fix $p=p_0$ and draw the contours of $\mi^{2d}(s,t;p_0,\mi):=\mi^{3d}(s,t,p_0;\mi)$ on slice $p=p_0$. 
\end{definition}

Figure~\ref{fig:contour} illustrates the contours for information-monotone MIs and a MI that is not information-monotone. 

\begin{figure}[!ht]

\includegraphics[width = \textwidth]{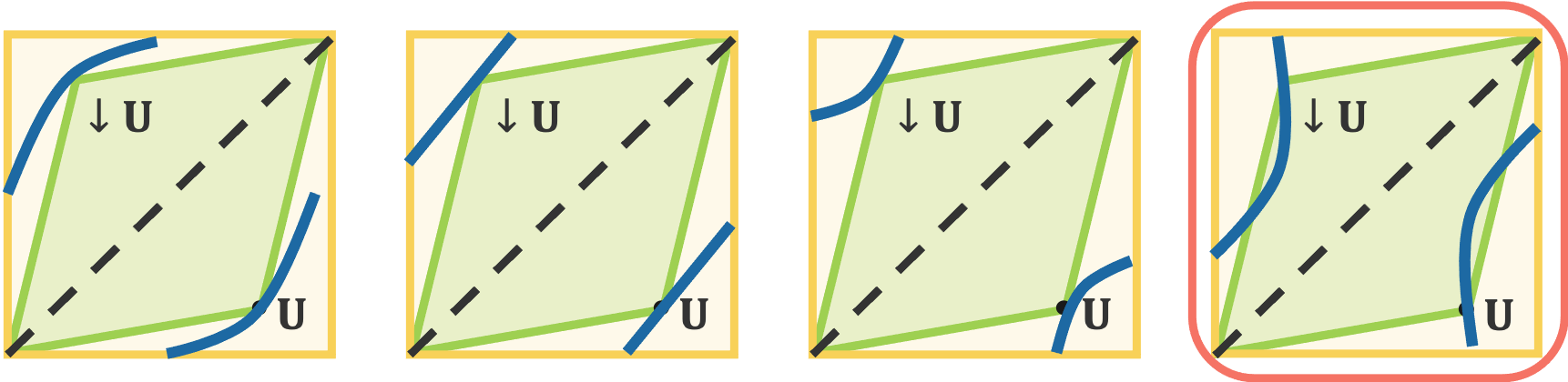}

\caption{\textbf{Information-monotone MI vs Un-information-monotone MI}: the first three figures illustrate the contours of different information-monotone MIs. In these figures, the contours on each element $U$ (the blue lines) must always contain $U$'s lower set (the green parallelogram). The last figure (with a red frame) illustrates the contours of a MI which is not information-monotone. } \label{fig:contour}
\end{figure}

We first visualize multiple commonly used MIs and compare their contours in the same square slice. 

\paragraph{Visualization of Commonly Used Mutual Information}\label{sec:exist}

We will visualize two existed commonly used mutual information measures in this section. These measures are designed by a distance-based approach. For two random variables $X$ and $Y$, $U_Y$ represents the prior distribution over $Y$ \emph{when we have no information}. That is $U_Y(y)=\Pr[Y=y]$. $U_{Y|x}$ denotes the posterior distribution $Y$, i.e. $U_{Y|x}(y)=\Pr[Y=y|X=x]$ \emph{when we have information $X=x$}. When $X$ and $Y$ are independent, knowing $X$ will not change our belief for $Y$, i.e., $U_{Y|X}$ equals $U_{Y}$. When $X$ and $Y$ are highly correlated, knowing $X$ changes the belief for $Y$ a lot, i.e., $U_{Y|X}$ is quite different from $U_{Y}$. Intuitively, we can use the ``distance'' between the informative prediction $U_{Y|X}$ and the uninformative prediction $U_{Y}$ to represent the mutual information between $X$ and $Y$. The distance measure should be picked carefully to satisfy information-monotonicity. Two different distance families, $f$-divergence $\mathrm{D}_f(\cdot,\cdot)$ and Bregman-divergence $\mathrm{D}_{PS}(\cdot,\cdot)$, can induce two families of information-monotone mutual information measures \cite{Kong:2019:ITF:3309879.3296670}. We list these measures here. 

\begin{itemize}

    \item $f$ Mutual Information ($\fmi^f$): $\E_{x\leftarrow U_X}\mathrm{D}_f(U_{Y|x},U_Y)$ 
    \item Bregman Mutual Information ($\bmi^{PS}$): $\E_{x\leftarrow U_X}\mathrm{D}_{PS}(U_{Y|x},U_Y)$ 

\end{itemize}

We then give two special cases of the above families. The commonly used KL-divergence belongs to both of the families and induces the classic Shannon mutual information. The commonly used scoring rule, the quadratic scoring rule, induces the quadratic mutual information. 

\begin{itemize}
\item Shannon Mutual Information (SMI): $\E_{x\leftarrow U_X}\mathrm{D}_{KL}(U_{Y|x},U_Y)$ 

\item Quadratic Mutual Information (QMI): $\E_{x\leftarrow U_X}||U_{Y|x}-U_Y||^2$ 

\end{itemize}

We visualize SMI and QMI, as well as DMI in Figure~\ref{fig:mi}.

\begin{figure}[!ht]
\begin{center}
\includegraphics[width = 15cm]{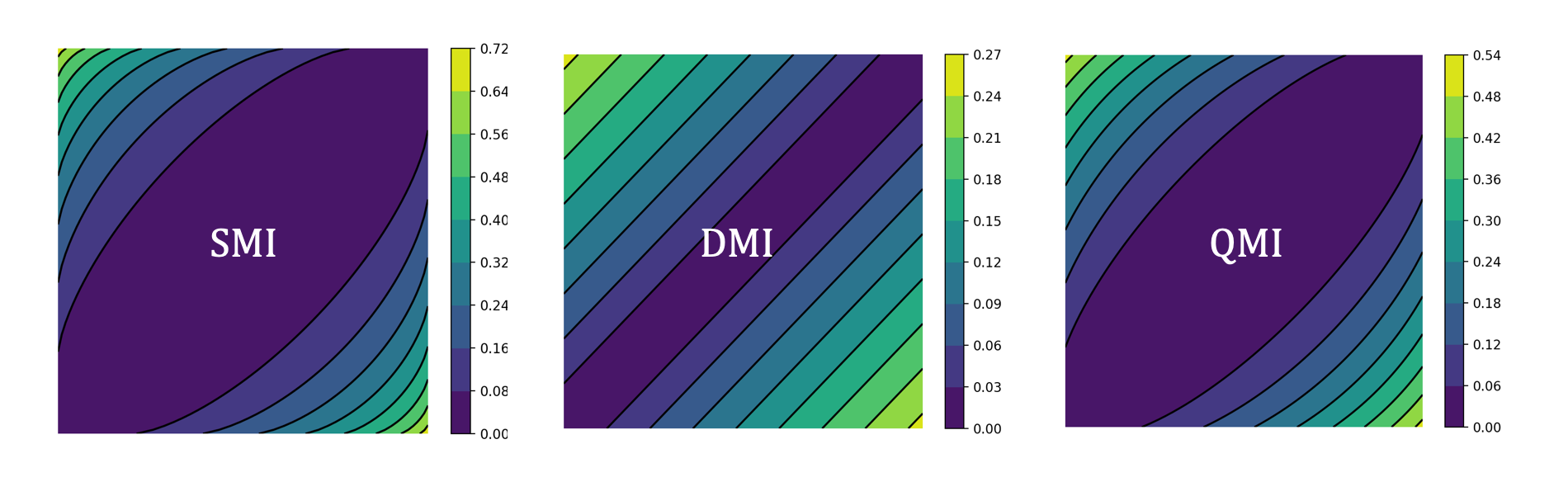}
\end{center}
\caption{Contours of SMI, DMI, QMI on slice $p=.5$: DMI has the parallel lines ``|||'' as contours and both SMI and QMI have shapes like ``(|)''. Compared with ``|||'', This ``(|)'' shape contour will punish two-sided noise (far from the square frame's boundary) more and one-sided noise (on the boundary of the square frame) less.}\label{fig:mi}
\end{figure}

\paragraph{Visualization of Binary Volume Information} We use the results of Lemma~\ref{lem:binaryvmi} and employ the computer to compute the indefinite integration and obtain the explicit formula of $\vmi^w$ in the binary case. 

\begin{lemma}\label{lem:binaryvmi}
In binary case, 
\begin{align*}
    \vmi^{w}(U)=& 2|\det(U)|\int_{s=0}^{1}\int_{t=0}^1 w(\begin{bmatrix} s & t\\ 1-s & 1-t\end{bmatrix} U )d s d t\\
    =& 2|u_{00}u_{11}-u_{10}u_{01}|\int_{s=0}^{1}\int_{t=0}^1 w(\begin{bmatrix} s & t\\ 1-s & 1-t\end{bmatrix} \begin{bmatrix}u_{00} & u_{01}\\ u_{10} & u_{11}\end{bmatrix}
 )d s d t
\end{align*}
\end{lemma}

We defer the proof to appendix. 

\begin{figure}[!ht]
\begin{center}
\includegraphics[width = \textwidth]{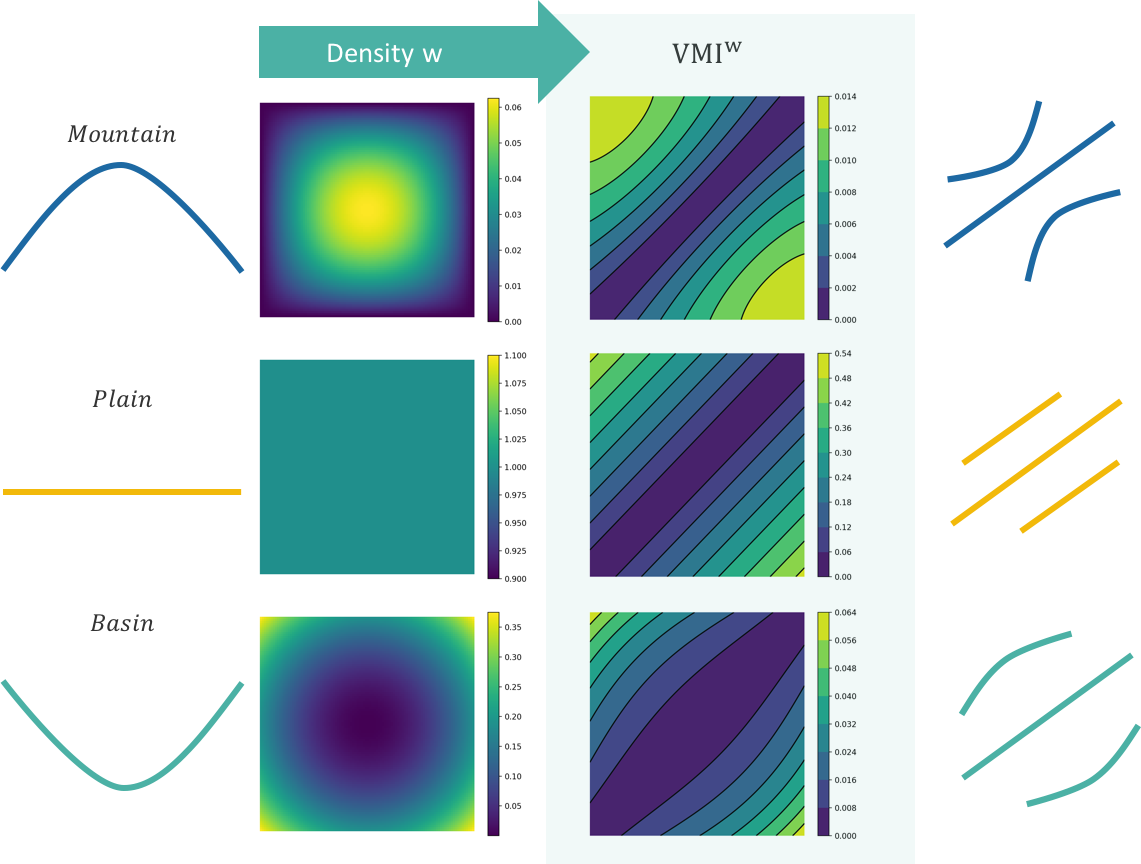}
\end{center}
\caption{\textbf{From density function $w$ to $\vmi^w$}: the left column shows the heatmaps of the density functions $w$ and the right column shows the contours of their corresponding $\vmi^w$s. The ``plain'' shape density has uniform density everywhere. Its corresponding VMI, DMI, has the parallel lines ``|||'' as contours. The ``mountain'' shape density has the highest density in the center. This will lead to a VMI whose contour is like ``)|(''. Compared with ``|||'', This ``)|('' shape contour will punish one-sided noise (e.g. say ``like'' when ``like'', say ``hate'' w.p. $\frac{1}{2}$ when ``hate'' ) more. The ``basin'' shape density has the lowest density in the center. This will lead to a VMI whose contour is like ``(|)''. Compared with ``|||'', This ``(|)'' shape contour will punish two-sided noise more. }\label{fig:wtovmi}

\end{figure}

\begin{example}\label{ex:binary}

Here we provide three concrete examples to show how the choice of density will affect the corresponding volume mutual information. 


We pick the $p_0=.5$ slice to illustrate the 2-dimensional contour of the VMIs, which is the contour of $\mi^{2d}(s,t;.5,\vmi^w)$. We will also draw the heatmap of the density function. In the $p_0=.5$ slice, in the new coordinates, the density function changes to $w^{2d}(s,t):=w(\begin{bmatrix} s & t\\ 1-s & 1-t\end{bmatrix}  \begin{bmatrix} .5 & 0\\ 0 & .5\end{bmatrix})$. 

\begin{enumerate}
    \item \textbf{Mountain} $w(\begin{bmatrix}a & b\\ c & d\end{bmatrix})=16 abcd$, $w^{2d}(s,t)=s(1-s)t(1-t)$: 

This density function is called ``Mountain'' since the center has a higher density than its surroundings. The highest density will be obtained when $s=t=.5$. \begin{align*}
    \vmi^{w}(U)
    = & 2|\det(U)|(\frac{8 u_{00}^{2}}{15} u_{01}^{2} + \frac{4 u_{01}}{3} u_{00}^{2} u_{11} + \frac{4 u_{00}^{2}}{9} u_{11}^{2}\\& + \frac{4 u_{00}}{3} u_{01}^{2} u_{10} + \frac{40 u_{00}}{9} u_{01} u_{10} u_{11} + \\& \frac{4 u_{00}}{3} u_{10} u_{11}^{2} + \frac{4 u_{01}^{2}}{9} u_{10}^{2} + \frac{4 u_{01}}{3} u_{10}^{2} u_{11} + \frac{8 u_{10}^{2}}{15} u_{11}^{2})\end{align*}
    \item \textbf{Plain} $w(\begin{bmatrix}a & b\\ c & d\end{bmatrix})=1$, $w^{2d}(s,t)=1$: \begin{align*}
    \vmi^{w}(U)=  & 2|\det(U)|\end{align*}
    \item \textbf{Basin} $w(\begin{bmatrix}a & b\\ c & d\end{bmatrix})=3((a-.25)^2+(b-.25)^2)$, $w^{2d}(s,t)=\frac{3}{4}((s-.5)^2+(t-.5)^2)$: 

This density function is called ``Basin'' since the center has a lower density than its surroundings. The lowest density will be obtained when $s=t=.5$.

\begin{align*}
    \vmi^{w}(U)=  & 2|\det(U)|(u_{00}^{2} + 1.5 u_{00} u_{10} + u_{01}^{2} + 1.5 u_{01} u_{11} \\& +  u_{10}^{2} +  u_{11}^{2}  - 0.375)\end{align*}
\end{enumerate}



The visualizations of $w$ and $\vmi^w$ are presented in Figure~\ref{fig:wtovmi}.

\end{example}

The above example also provides three concrete polynomial mutual information by multiplying $|\det(U)|$ to each of them. The plain one corresponds to DMI's square while the mountain and basin density provide two new polynomial mutual information for the binary case, which leads to two new practical dominantly truthful peer prediction mechanisms. 

\paragraph{Visualization of a New Practical Dominantly Truthful Mechanism} \label{sec:newmec}
We have proved that polynomial VMI can be used to construct practical mechanism. Here we will also provide a concrete example in the binary case. 
We use a new polynomial binary mutual information $\vmi^{\star}$ to construct a new peer prediction mechanism in the binary case. Our results work for non-binary case, this example uses the binary case for ease of illustration. 
We pick the ``Mountain'' case (Example~\ref{ex:binary}) and multiply $\det(U)$ to obtain a new polynomial binary mutual information $\vmi^{\star}$.

\begin{align*}
    \vmi^{\star}(X;Y)
    = & 2(u_{00}u_{11}-u_{01}u_{10})^2 (\frac{8 u_{00}^{2}}{15} u_{01}^{2} + \frac{4 u_{01}}{3} u_{00}^{2} u_{11} + \frac{4 u_{00}^{2}}{9} u_{11}^{2}\\& + \frac{4 u_{00}}{3} u_{01}^{2} u_{10} + \frac{40 u_{00}}{9} u_{01} u_{10} u_{11} + \\& \frac{4 u_{00}}{3} u_{10} u_{11}^{2} + \frac{4 u_{01}^{2}}{9} u_{10}^{2} + \frac{4 u_{01}}{3} u_{10}^{2} u_{11} + \frac{8 u_{10}^{2}}{15} u_{11}^{2})\end{align*}

where $U=\begin{bmatrix}
    u_{00} & u_{01}\\u_{10}&u_{11}
\end{bmatrix}$ is the joint distribution matrix of $X,Y$.

It's hard to tell that $\vmi^{\star}$ satisfies the information-monotonicity from the above formula while $\vmi^{\star}$'s contour plot (Figure~\ref{fig:demo}) intuitively shows the monotonicity. Section~\ref{sec:binary} shows that $\vmi^{\star}$ is information-monotone from its construction. With $\vmi^{\star}$'s formula, we can construct a new constant-round dominantly truthful mechanism in the binary case by paying the participants the unbiased estimator of $\vmi^{\star}$. Previously, DMI-Mechanism is the only known constant-round dominantly truthful mechanism. 
    
\paragraph{$\vmi^{\star}$-Mechanism} $n$ participants are assigned $T\geq 8$ a priori similar tasks. The participants finish the tasks without any communication. 
\begin{description}
\item[Report] For each task $t$, each participant $i$ privately receives $c_i^t$ and reports $\hat{c}_i^t$. 
\item[Payment] For every two agents $i\neq j\in [n]$, we arbitrarily pick $8$ tasks and $E_t(c,c')$ is a binary indicator event such that $E_t(c,c')=1$ if for task $t$, agent $i$'s answer is $c$ and agent $j$'s answer is $c'$. Otherwise, $E_t(c,c')=0$. We define
 
\begin{align*}
    p_{ij}
    := & 2(E_1(0,0)E_2(1,1)-E_1(0,1)E_2(1,0)) (E_3(0,0)E_4(1,1)-E_3(0,1)E_4(1,0)) \\
    &\bigg(\frac{8 E_5(0,0)E_6(0,0)}{15} E_7(0,1)E_8(0,1)+ \frac{4 E_5(0,1)}{3} E_6(0,0) E_7(0,0) E_8(1,1)\\
    & + \frac{4 E_5(0,0)E_6(0,0)}{9} E_7(1,1)E_8(1,1) + \frac{4 E_5(0,0)}{3} E_6(0,1)E_7(0,1) E_8(1,0) \\&+ \frac{40 E_5(0,0)}{9} E_6(0,1) E_7(1,0) E_8(1,1)+ \frac{4 E_5(0,0)}{3} E_6(1,0) E_7(1,1)E_8(1,1)\\& + \frac{4 E_5(0,1)E_6(0,1)}{9} E_7(1,0)E_8(1,0) + \frac{4 E_5(0,1)}{3} E_6(1,0)E_7(1,0) E_8(1,1) \\&+ \frac{8 E_5(1,0)E_6(1,0)}{15} E_7(1,1)E_8(1,1)\bigg)\end{align*}
Agent $i$'s payment is $p_i:= \sum_{j\neq i\in [n]} p_{ij} $

\end{description}

The above mechanism is a special mutual information paradigm by using $\vmi^{\star}$'s unbiased estimator. According to Lemma~\ref{lem:paradigm}, $\vmi^{\star}$-Mechanism is dominantly truthful, prior-independent and works for $\geq 8$ tasks.

\begin{figure}[!ht]     \begin{center}
        \includegraphics[width=9cm]{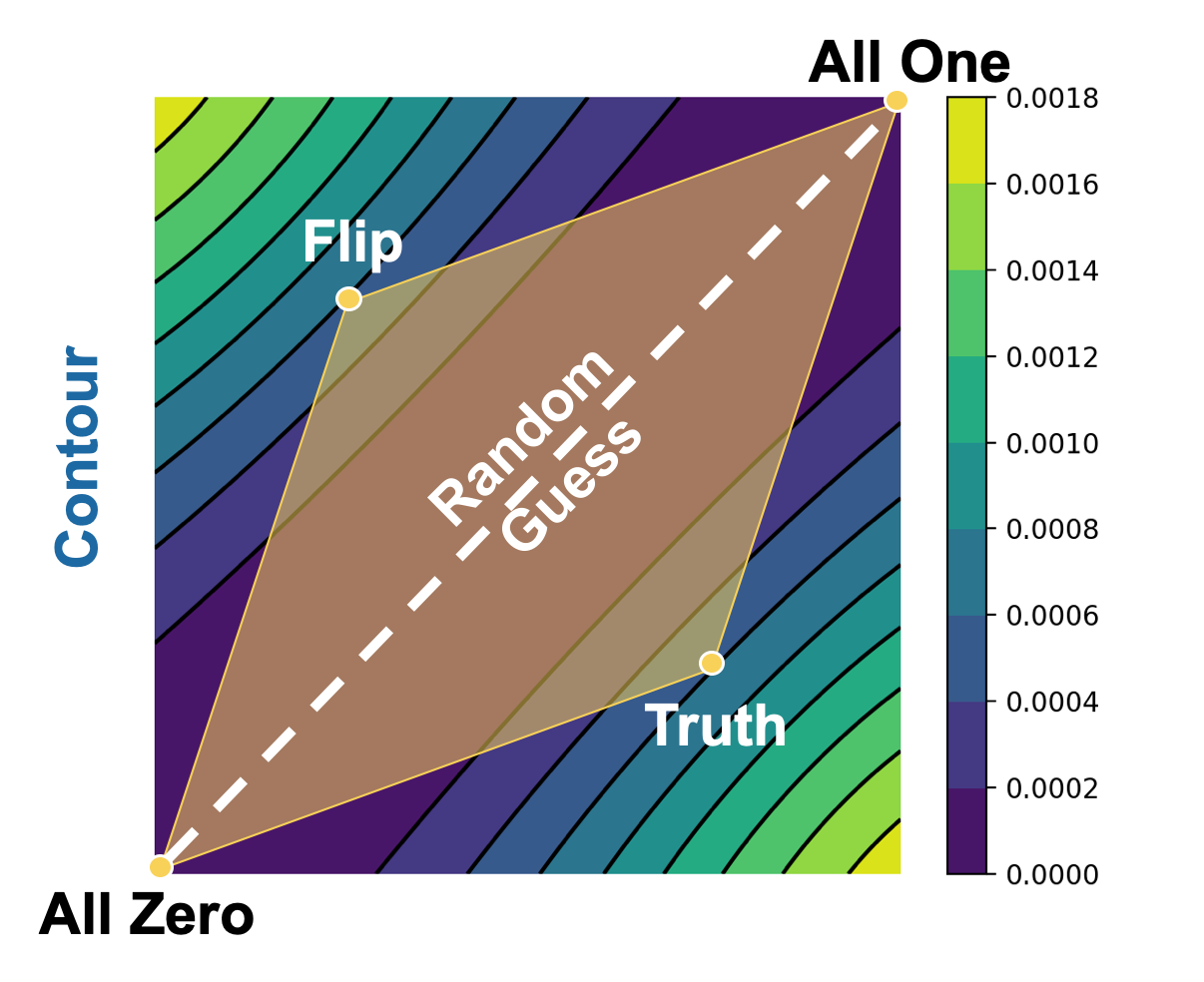}
        \caption{Illustration for $\vmi^{\star}$-Mechanism: Alice and Bob participate in the mechanism. Fixing Bob's strategy, when $U$ is the joint distribution over Bob and honest Alice's reports, Alice's strategy $\mathbf{S}$ corresponds to joint distribution $\mathbf{S}U$. We draw the contours of $\vmi^{\star}$ on the slice on $U$ and visualize Alice's strategy simultaneously. All strategies consist of a light yellow parallelogram with four pure strategies as vertices: truth-telling, always flipping the answer, always answering zero, always answering one. From the plot, when Alice tells the truth or always flips her answer, she will be paid the highest. When Alice reports uninformative answer like always saying zero/one or random guessing without looking at the questions, she will be paid zero, i.e., the lowest.} \label{fig:demo}

    \end{center}
\end{figure}

\section{Optimizing Multi-task Peer Prediction}\label{sec:opt}

Finally this section will discuss the optimization of multi-task peer prediction and use VMI to construct the optimal multi-task peer prediction. 

We start by introducing the optimization goal. The dominant truthfulness guarantees that truth-telling is the best report strategy, given that the participants receive the signals, while it may not give the participants incentive to spend a sufficient amount of effort to perform the tasks. Most previous work's analysis focuses on the setting where the participants do not need to invest any effort to obtain the signals (e.g. Do you like Panda Express). In this case, dominant truthfulness is sufficient. However, for a certain amount of tasks (e.g. online product evaluation, art evaluation), the participants need to invest effort. We will introduce an effort strategy model such that we can properly define the mechanism design goal about incentivizing efforts. 

\paragraph{Effort Strategy Model} We assume that when Alice and Bob spend full efforts, the joint distribution over their signals is $U_G$. Alice can pick an effort strategy that leads to an intrinsic noise $N_A\in \mathbb{R}^{C\times C}$ for the signal she observes. That is, $N_A(c',c)$ is the probability that her full effort's signal is $c$, while she observes signal $c'$. Alice's effort is modeled as a function of her intrinsic noise $N_A\in \mathbb{R}^{C\times C}$, $e_A(N_A)$. 

The requester's expected \emph{value} for the elicited answers is a function of the underlying joint distribution over Alice and Bob's answers, $v(\hat{U}_{A,B})$. In our model, since both $U_{A,B}$ and $\hat{U}_{A,B}$ can be seen as $C\times C$ matrices, we can represent $\hat{U}_{A,B}$ as follows. \[\hat{U}_{A,B}=S_A U_{A,B} S_B ^{\top}=S_A N_A U_G N_B^{\top}S_B ^{\top}.\] 

We will make natural monotonicity and continuity assumptions for the value and effort functions. Intuitively, more noisy intrinsic noise requires less effort and leads to less value to the task requester. 

\begin{assumption} [Information-monotonicity, continuous value/effort, and finite effort level choices]\label{assume:mono}
We assume that the effort functions and value functions and information-monotone in the sense that 
\[ \forall U'\preceq U, v(U')\leq v(U), v(U^{'\top})\leq v(U^{\top});\] \[\forall\footnote{We can naturally extend Definition~\ref{def:lpo} to the space of all column-stochastic matrices, i.e, effort strategies. That is, for two column-stochastic matrices $N'\preceq N$ if there exists a column-stochastic matrix $T$ such that $N'= T N$.} N'\preceq N, e_A(N')\leq e_A(N), e_B(N') \leq e_B(N) \] which implies that  post-processing the data does not require any effort or increase the value. We additionally assume that the value/effort function is continuous and both Alice and Bob pick their effort strategies from a finite discrete set.

\end{assumption}\label{assume:opt}

We will optimize over dominantly truthful and practical mechanisms. Thus, once Alice and Bob determine their effort strategies, they will truthfully report their signals. Therefore, we can use $U_{A,B}=N_A U_G N_B^{\top}$ instead of $\hat{U}_{A,B}$. Then Alice's expected payment is a function of $U_{A,B}$ and denoted by $\mathcal{P}_A(U_{A,B})$. We model Bob analogously.

\begin{example}\label{eg:opt}
Alice and Bob are assigned multiple similar quality evaluation task. Alice has three possible effort strategies which lead to the following intrinsic noises: 
\begin{align*}
&N_A^0(\text{bad},\text{bad})=.5,N_A^0(\text{bad},\text{good})=.5, e_A(N_A^0)=0\tag{full noise}\\ 
&N_A^1(\text{bad},\text{bad})=1,N_A^1(\text{bad},\text{good})=.4, e_A(N_A^1)=1\tag{one-sided noise}\\
&N_A^2(\text{bad},\text{bad})=.8,N_A^2(\text{bad},\text{good})=.2, e_A(N_A^2)=10 \tag{two-sided noise}
\end{align*}
and Bob has two possible effort strategies which lead to intrinsic noises $N_B^0=N_A^0,e_B(N_B^0)=0$, $N_B^1=N_A^1,e_B(N_B^1)=1$. 

Regarding the value of the requester, when either Alice or Bob's signal is fully noisy, the requester's value will be zero. For other cases, 
\[v(N_A^1 U_G N_B^{1\top})=15, v(N_A^2 U_G N_B^{1\top})=50\]


\end{example}

\paragraph{Optimization Goal} The optimization problem is \begin{align*}
\max_{\mathcal{P}_A,\mathcal{P}_B} \quad& v(U_{A,B})-\mathcal{P}_A(U_{A,B})-\mathcal{P}_B(U_{A,B}) \tag{maximize the requester's expected utility}\\
\text{s.t.} \quad& U_{A,B}=N_A U_G N_B^{\top}\\\quad& N_A\in\arg\max_{N_A'} \mathcal{P}_A(N_A' U_G \tag{$(N_A, N_B)$ consists of an equilibrium} N_B^{\top})-e_A(N_A')\\
& N_B\in\arg\max_{N_B'} \mathcal{P}_B(N_A U_G N_B'^{\top})-e_B(N_B')
\end{align*}

If there are multiple equilibria $(N_A, N_B)$, Alice and Bob will choose the equilibrium that maximizes min(Alice's expected utility, Bob's expected utility). If there are multiple equilibria that maximize their min expected utility, we will maximize the lower bound of the requester's utility over those equilibria.

\paragraph{DMI is not optimal} In this example, $N_A^1=\begin{bmatrix}1&.4\\0&.6\end{bmatrix}$ and $N_A^2=\begin{bmatrix}.8&.2\\.2&.8\end{bmatrix}$ have the same determinant, thus, \dmi-Mechanism must reward Alice the same amount of payment no matter Alice pick the one-sided noise effort or two-sided noise effort. Then as long as the expected payment is greater than 1, Alice must pick the one-sided noise since it requires much less effort. However, the requester values the other choice, the two-sided one, much more even if the requester should pay more. Later we will show, unlike DMI-mechanism which is less pleasant to the requester in this setting, a series of VMI-mechanisms can approximately make the requester obtain the optimal utility.

\paragraph{Modeling discussion} This optimization goal requires the knowledge of $U_G$ and the cost of different effort strategies. Note that $U_G$ does not represent the full knowledge. For example, the requester knows that about $10\%$ products are bad thus $U_G=\begin{bmatrix}
	10\% & 0 \\
	0 & 90\%
\end{bmatrix}$. However, the requester does not know which products are bad, thus she still need to elicit information from the crowds. The cost of different effort strategies represents the requester's estimation for the task difficulty. For example, for some tasks it may be easy to get a 80\% accurate answer but very difficult to get a 90\% accurate answer. Though this optimization goal requires a certain prior knowledge, we believe this gives the first step for effort incentive optimization over practical multi-task peer prediction mechanisms.

We will optimize over all possible $\mathcal{P}_A,\mathcal{P}_B$ which are Alice and Bob's \emph{expected} payments under \emph{dominantly truthful and practical} mechanisms. That is why the above formula does not involve Alice and Bob's report strategies. After we find a family of dominantly truthful and practical mechanisms, we can directly optimize the above goal over the family. Another way is to first optimize over all possible dominantly truthful $\mathcal{P}_A,\mathcal{P}_B$, even if there does not exist a practical mechanism which pays $\mathcal{P}_A,\mathcal{P}_B$ in expectation\footnote{In other words, we can implement such $\mathcal{P}_A,\mathcal{P}_B$ only if we have the perfect estimation of $\hat{U}_{A,B}$ from infinite number of tasks.}. Then we can use a sequence of practical mechanisms to approximate the optimal dominantly truthful mechanism. It turns out the second approach is much easier in our setting. 

\begin{description}
\item [Step 1 Practical VMI-Mechanisms:] Generalize DMI-Mechanism to a family of dominantly truthful and practical mechanisms, VMI-Mechanisms;
\begin{description}
\item [Step 1.1 Mechanism design $\Rightarrow$ Mutual information design:] Reduce the design of dominantly truthful and practical mechanisms to the design of polynomial information-monotone mutual information measure (Section~\ref{sec:mip});
\item [Step 1.2 VMI construction:] Construct information-monotone Volume Mutual Information (VMI) and show that we can obtain polynomial VMI by assigning distribution space a polynomial density (Section~\ref{sec:vmi});
\end{description}
\item [Step 2 Optimal threshold payment:] Optimize over all possible dominantly truthful $\mathcal{P}_A,\mathcal{P}_B$ and show that the optimal payment function is a threshold function (Section~\ref{sec:theta});
\item [Step 3 Approximating threshold payment via VMI-Mechanisms:] Show that the optimal threshold payment corresponds to a special VMI with Dirac delta density; use a sequence of polynomial densities to approximate the Dirac delta density and finally construct corresponding VMI-Mechanisms (Section~\ref{sec:approx}). 
\end{description}

We have finished the first step and will start the next two steps. 

\subsection{Optimal Threshold Payment}\label{sec:theta}

We will show that the optimal expected payment function is a threshold function. First, we observe that the requester should pay at least the participants' efforts. Thus, in the above example, the requester's utility will be at most either $15-1-1=13$ or $50-10-1=39$. This observation is formalized as follows.

\begin{observation}
The requester expected utility must be less than \[  v^*:=\max_{N_A,N_B} \quad v(U_{A,B})- e_A(N_A)-e_B(N_B), U_{A,B}=N_A U_G N_B^{\top}\]
\end{observation}

\begin{proof}
The participants are willing to participate if and only if their expected utility is positive. In such case, 
\begin{align*}
&\mathcal{P}_A(N_A U_G N_B^{\top})>e_A(N_A)\\
&\mathcal{P}_B(N_A U_G N_B^{\top})>e_B(N_B)
\end{align*}
Thus, the requester's utility is less than 
\[v(U_{A,B})-e_A(N_A)- e_B(N_B)\] which is less than 
\[  \max_{N_A,N_B} \quad v(U_{A,B})- e_A(N_A)-e_B(N_B), U_{A,B}=N_A U_G N_B^{\top}.\]
\end{proof}

We pick optimal $U^*,N_A^*,N_B^*$ such that $U^*=N_A^* U_G N_B^{*\top}$ where \[N_A^*,N_B^*\in \arg\max_{N_A,N_B} \quad v(U_{A,B})- e_A(N_A)-e_B(N_B).\] We assume that $U^*$ is \emph{non-degenerate}, i.e., $\det(U^*)\neq 0$. By setting $U^*$ as a threshold and just pay the efforts participants make will be optimal. The result is formalized in the following proposition. In the above example, the threshold can be set as $N_A^2 U_G N_B^{1\top}$. This guarantees that Alice picks the desired two-sided noise effort and the requester will obtain the optimal utility 50-1-10=39. 

\begin{proposition}
For all $\epsilon>0$, by setting Alice's expected payment function as \[\mathcal{P}_A(U)=(e_A(N_A^*)+\epsilon)\mathbbm{1}(U\succeq U^*)\] and Bob's expected payment function as \[\mathcal{P}_B(U)=(e_B(N_B^*)+\epsilon)\mathbbm{1}(U^{\top}\succeq U^{*\top}),\] the requester will obtain at least an almost optimal utility $v^*-2\epsilon$ and the payment is bounded by $e_A(N_A^*)+e_B(N_B^*)+2\epsilon$. 
\end{proposition}
\begin{proof}
First, if Alice and Bob choose an equilibrium where their joint distribution $U$ does not satisfy  $U\succeq U^*$ or $U^{\top}\succeq U^{*\top}$, then one of them will obtain 0 utility. Thus, if there exist equilibria that lead to $U\succeq U^*,U^{\top}\succeq U^{*\top}$ and both of them can obtain a strictly positive utility, then Alice and Bob must pick one of such equilibria, since we assume they will pick the equilibrium that maximizes min(Alice's expected utility, Bob's expected utility). We will show that such equilibrium exists by showing that $N_A^*,N_B^*$ is such an equilibrium. When Bob plays $N_B^*$, to obtain a strictly positive utility, Alice must play an effort strategy $N_A$ such that the corresponding joint distribution is more informative than $U^*$ from Alice's side, that is, $N_A U_G N_B^*\succeq U^*=N_A^* U_G N_B^*$. Then there exists $T$ such that $T N_A U_G N_B^*=N_A^* U_G N_B^*$ which implies that $T N_A=N_A^*$ since $U^*$ is non-degenerate. Thus, we have $N_A\succeq N_A^*$. Since the effort function is monotone, $e_A(N_A)\geq e_A(N_A^*)$, $N_A^*$ is a best effort strategy for Alice when Bob plays $N_B^*$. The analysis for Bob's side is analogous. Thus, $(N_A^*,N_B^*)$ is an equilibrium. Combining the above analysis, the requester's utility will be at least $v(U^*)-e_A(N_A^*)-e_B(N_B^*)-2\epsilon= v^*-2\epsilon$. 
\end{proof}

However, there does not exist any finite-number-of-tasks mechanism which pays the above optimal threshold function in expectation. Therefore, we will approximate this function by a series of polynomial mutual information and then employ mutual information paradigm to construct the corresponding mechanism. 

A naive attempt is to compute the polynomial approximation of the threshold function directly. However, the obtained polynomial approximation may not be information-monotone thus cannot induce the dominant truthfulness. Thus, instead of computing a polynomial approximation of the threshold function directly, we will compute a polynomial approximation of the ``derivative'' of the threshold function, a Dirac delta function, and then use the approximation as density to construct the corresponding polynomial volume mutual information. 


\subsection{Approximating the Optimal Threshold Payments via VMI-Mechanisms}\label{sec:approx}

To formally state our approximation process, we first give a formal definition for slice whose intuition has been illustrated in Figure~\ref{fig:binary}.

\begin{definition}[Slice]
For all joint distribution $U$, we define $u_j:=\sum_{i} u_{ij}$ as the the sum of the $j^{th}$ column of $U$. We use $slice(U)$ to define the space of joint distributions whose column sums are the same as $U$. That is, $slice(U):=\{U'|\forall j, u'_j=u_j\}$. In Figure~\ref{fig:binary}, $slice(U)$ is the slice that contains $U$. 
\end{definition}

We define a Dirac delta function $\delta_{U^*}$ such that $\forall U\neq U^*, \delta_{U^*}(U)=0$, for all open set $O\subset slice(U^*)$ that contains $U^*$, $\int_{O} \delta_{U^*}(U)d U = 1$. When we use the Dirac delta function as density, the corresponding VMI will be a threshold function $\mathbbm{1}(U\succeq U^*)$.

\paragraph{Approximation of the Optimal Payment}
The process has three steps. 
\begin{description}
\item [Step 1: Polynomial approximation for Dirac delta] Given degree $k$, we obtain an polynomial approximation $\phi_k,\psi_k$ for $\delta_{U^*},\delta_{U^{*\top}}$
\item [Step 2: Using the polynomial density to construct VMI] When $C$ is even, we set the expected payment function as \[\mathcal{P}_A(U)=(e_A(N_A^*)+\epsilon)(\vmi^{\phi_k}(U))^2,\mathcal{P}_B(U)=(e_B(N_B^*)+\epsilon)(\vmi^{\psi_k}(U^{\top}))^2,\] when $C$ is odd, 
we set the expected payment function as \[\mathcal{P}_A(U)=(e_A(N_A^*)+\epsilon)\vmi^{\phi_k}(U),\mathcal{P}_B(U)=(e_B(N_B^*)+\epsilon)\vmi^{\psi_k}(U^{\top}).\] 
\item [Step 3: Constructing the VMI-Mechanisms] We set the mechanism correspondingly: when $C$ is even, \[p_A=(e_A(N_A^*)+\epsilon)\ube^{(\vmi^{\phi_k})^2}(\{(\hat{c}_A^t,\hat{c}_B^t)\}_{t=1}^{T}),\mathcal{P}_B(U)=(e_B(N_B^*)+\epsilon)\ube^{(\vmi^{\psi_k})^2}(\{(\hat{c}_B^t,\hat{c}_A^t)\}_{t=1}^{T});\] when $C$ is odd, \[p_A=(e_A(N_A^*)+\epsilon)\ube^{\vmi^{\phi_k}}(\{(\hat{c}_A^t,\hat{c}_B^t)\}_{t=1}^{T}),\mathcal{P}_B(U)=(e_B(N_B^*)+\epsilon)\ube^{\vmi^{\psi_k}}(\{(\hat{c}_B^t,\hat{c}_A^t)\}_{t=1}^{T}).\]
\end{description}

We use the square of $(\vmi^{\phi_k}(U))^2$ when $C$ is even to guarantee that it is a polynomial based on the results of Theorem~\ref{thm:vmi}. It's left to construct a polynomial approximation for Dirac delta density. We will use a Dirichlet distribution family-based VMI to construct the polynomial distribution. As we mentioned before, we are inspired by a beta family of scoring rules \cite{2005Loss,2013Choosing} which are used to approximate a threshold scoring rule, ``misclassification'' scoring.

\paragraph{Dirichlet/Multivariate Beta distribution} We first introduce Dirichlet distributions. \begin{definition}[Dirichlet distribution~\cite{2006Continuous}]
Given $K\geq 2$, for all parameters $\bm{\beta}=\beta_1,\beta_2,\cdots,\beta_K>0$, the Dirichlet distribution $Dir(\beta_1,\beta_2,\cdots,\beta_K)$ is defined as a continuous multivariate probability distribution with density \[w^{\bm{\beta}}(x_1,x_2,\cdots,x_K)=\frac{1}{B(\bm{\beta})}\Pi_i x_i^{\beta_i-1}\] with respect to Lebesgue measure on $\mathbb{R}^{K-1}$ where $\sum_k x_k=1,x_k\geq 0$ and $B(\cdot)$ is the beta function. 
\end{definition}

\begin{fact}[Mean/Variance of Dirichlet-distributed variables~\cite{2006Continuous}]\label{fact:beta}
For $Dir(\beta_1,\beta_2,\cdots,\beta_K)$-distributed random variables $(X_1,X_2,\cdots,X_K)$, the mean of $X_k$ is $\frac{\beta_k}{\sum_i \beta_i }$ and the variance of $X_k$ is $\frac{\frac{\beta_k}{\sum_i \beta_i }(1-\frac{\beta_k}{\sum_i \beta_i })}{\sum_i \beta_i+1}$.
\end{fact}

\paragraph{Dirichlet family of VMI} We define a parametric family of VMI where the density function is inspired from Dirichlet distributions. 

\begin{definition}[Dirichlet family of VMI]
We define the Dirichlet family of volume mutual information, $\vmi^{w^{\bm{\alpha}}}$, by parameterizing the density function as  
\[w^{\bm{\alpha}}(U)=\frac{1}{C(\bm{\alpha})}\Pi_{ij} u_{ij}^{\alpha_{ij}-1}\] regarding parameters $\bm{\alpha}:=\{\alpha_{ij}>0,i,j\in[C]\}$ and $C(\bm{\alpha})$\footnote{In fact, the proof of Lemma~\ref{lem:dirich} shows that $C(\bm{\alpha})\propto\Pi_j {\alpha_j}^{\alpha_j-1}*B(\bm{\alpha}_j)$ where $\alpha:=\sum_{ij}\alpha_{ij}, \alpha_j:=\sum_i \alpha_{ij},\bm{\alpha}_j=(\alpha_{1j},\alpha_{2j},\cdots)$.} is a normalizing constant such that the volume of $slice(\bm{\alpha}/\alpha)$ be 1, i.e., $\int_{x\in slice(\bm{\alpha}/\alpha)} w^{\bm{\alpha}}(x) d\mathscr{H}^{C(C-1)}(x)=1$. 
\end{definition}



\begin{lemma}\label{lem:dirich}
Given a joint distribution $U^{\star}$, we set $\bm{\alpha}(U^{\star})=\{\alpha u^{\star}_{ij},i,j\in[C],\alpha>0\}$, for all $U$ such that $U^{\star} \notin \partial{\downarrow U}$, 
\[ \lim_{\alpha\rightarrow \infty} \vmi^{w^{\bm{\alpha}(U^{\star})}}(U)=\mathbbm{1}(U\succeq U^{\star})\] where $\partial{\downarrow U}$ is the boundary of $\downarrow U$.
\end{lemma}

To prove the above lemma, we first observe that for $U\notin slice(U^{\star})$ (which is definitely not more informative than $U^{\star}$), VMI at $U$ is less than the volume of $slice(U)$. Then we will show that the volume of $slice(U)$ goes to zero, which implies that VMI at $U$ goes to zero. For joint distribution on $slice(U^{\star})$, we will show that the density function restricted to $slice(U^{\star})$ is a probability density over $C$ independent Dirichlet-distributed random variables and we can show that it converges in distribution to constant $U^{\star}$ at continuous point, which leads to the above lemma's results. We defer the formal proof to Appendix~\ref{sec:additionalproof}. 

Note that the convergence happens only for $U$ whose lower set's boundary does not contain the special $U^{\star}$ such that it has zero measure in the limit to guarantee continuity. Then if we set $U^{\star}=U^*$ directly, the VMI at $U^*$ will not converge to 1. Thus, instead, we will use a lower-bound of $U^*$ as a substituted threshold such that the VMI at $U^*$ converges to one and the requester's value will only be sacrificed a little bit by using $\mathbbm{1}(U\succeq U^{\star})$ instead of $\mathbbm{1}(U\succeq U^*)$.

We formally state the polynomial approximation part here. Figure~\ref{fig:optapprox} presents an illustration. 

\begin{figure}[!ht]
\begin{center}
\includegraphics[width = 10 cm]{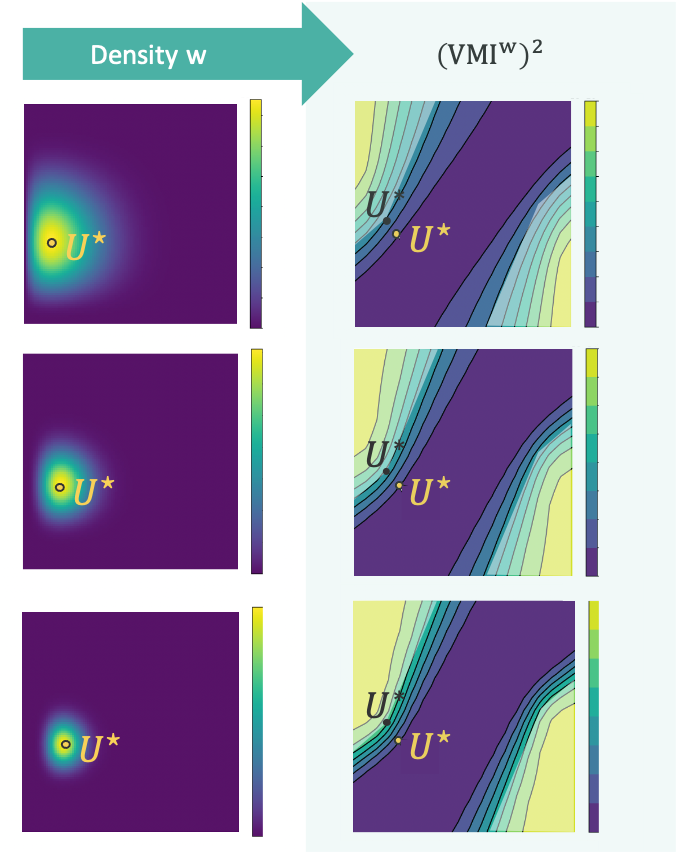}
\end{center}
\caption{\textbf{Polynomial approximation of OPT threshold function}: We aim to approximate the optimal $\mathbbm{1}(U\succeq U^*)$, which is shown as the light white area in the figures. To guarantee that convergence happens at $U^*$, we use $U^*$'s substituted threshold $U^{\star}=\begin{bmatrix}.2&.1\\.3&.4\end{bmatrix}$. The density is $w^{\alpha(U^{\star})}(U)\propto u_{00}^{.2\alpha -1}u_{01}^{.1\alpha -1}u_{10}^{.3\alpha -1}u_{11}^{.4\alpha -1}$ and $\alpha=20,50,100$ from top to bottom. The right side shows the contours of $(\vmi^w)^2$, which are guaranteed to be polynomial. As $\alpha$ increases, the density becomes more concentrated on $U^{\star}$ and the corresponding volume mutual information becomes closer to the optimal $\mathbbm{1}(U\succeq U^*)$. }\label{fig:optapprox}

\end{figure}

\begin{description}
\item [Polynomial approximation for Dirac delta] Given the optimal $U^*$, we pick a non-degenerate $N_A^{\star}U_G N_B^{*}=:U^{\star}\prec U^*=N_A^{*}U_G N_B^{*}$ such that $v(U^{\star})\geq v(U^*)-\epsilon$, $e_A(N_A^{\star})\geq e_A(N_A^{*})-\epsilon$, $U^{\star}\notin \partial{\downarrow U^*}$, and all numbers in $U^{\star}$ are rational, as the substituted threshold. Given proper integer $\alpha>0$ such that all $\{\alpha U^{\star}_{ij},i,j\in[C]\}$ are integers, we use $\phi_\alpha:=w^{\bm{\alpha}(U^{\star})}$ as our polynomial approximations. We define $\psi_\alpha$ analogously.
\end{description}

We can always find such $U^{\star}$ since we assumed that $U^*$ is non-degenerate and the value/effort function is continuous. We will use the above polynomial approximation to construct the corresponding VMI, as well as the VMI-Mechanism. To have the result that sufficiently large $\alpha>0$ will lead to an almost optimal utility for the requester, we need to relax the equilibrium requirement for the effort strategy profile to $\delta$-equilibrium in the optimization goal.

\begin{definition}[$\delta$-equilibrium]
A strategy profile is a $\delta$-equilibrium if for each agent, given other agent's strategy, she cannot change her strategy to improve her expected utility by more than $\delta$. 
\end{definition}

The relaxation guarantees that in the mechanism which approximately pays participants in a threshold manner in expectation, the effort strategy profile $(N_A^*,N_B^*)$ at the threshold can still be considered by the participants. Note that we do not need any relaxed solution concept for agents' report strategies.

\begin{theorem}
For all $\delta>\epsilon>0$, there exists sufficiently large $\alpha>0$ such that when $C$ is even (odd), $T\geq 2(\alpha-C)$ ($T\geq \alpha-C$), mechanism  \[p_A=(e_A(N_A^*)+\epsilon)\ube^{(\vmi^{\phi_\alpha})^2}(\{(\hat{c}_A^t,\hat{c}_B^t)\}_{t=1}^{T}),\mathcal{P}_B(U)=(e_B(N_B^*)+\epsilon)\ube^{(\vmi^{\psi_\alpha})^2}(\{(\hat{c}_B^t,\hat{c}_A^t)\}_{t=1}^{T})\] \[\left(p_A=(e_A(N_A^*)+\epsilon)\ube^{\vmi^{\phi_\alpha}}(\{(\hat{c}_A^t,\hat{c}_B^t)\}_{t=1}^{T}),\mathcal{P}_B(U)=(e_B(N_B^*)+\epsilon)\ube^{\vmi^{\psi_\alpha}}(\{(\hat{c}_B^t,\hat{c}_A^t)\}_{t=1}^{T})\right)\] is practical, dominantly truthful and prior-independent. If we relax the equilibrium requirement for effort strategy profile to $\delta$-equilibrium in the optimization goal, the requester can obtain at least an almost optimal utility $v^*-4\epsilon$.
\end{theorem}

\begin{proof}

To distinguish, we denote the substituted threshold $U^{\star}$ for Alice's (Bob's) side as $U_A^{\star}$ ($U_B^{\star}$). First, if Alice and Bob choose an equilibrium where their joint distribution $U$ does not satisfy  $U\succeq U_A^{\star}$ or $U^{\top}\succeq U_B^{\star\top}$, then one of them's expected payment will converge to zero due to Lemma~\ref{lem:dirich} (note that if $U\nsucceq U_A^{\star}$, then we must have $U_A^{\star}\notin \partial{\downarrow U}$ such that the convergence happens at $U$; Bob's side is analogous). Moreover, both Alice and Bob pick their effort strategies from a finite discrete set. Thus, if there exist equilibria that lead to $U\succeq U_A^{\star}, U^{\top}\succeq U_B^{\star\top}$ and both of them can obtain a strictly positive utility in the limit, then Alice and Bob must pick one of such equilibria with sufficiently large $\alpha$, since we assume they will pick the equilibrium that maximizes min(Alice's expected utility, Bob's expected utility). We will show that such equilibrium exists when we relax to $\delta$-equilibrium. In fact, we will show that $(N_A^*,N_B^*)$ is such a $\delta$-equilibrium. 

Note that except $slice(U_A^{\star})$, other slices' volume will go to zero when $\alpha$ goes to infinity (Lemma~\ref{lem:dirich}). Moreover, the volume of $slice(U_A^{\star})$ is one due to our definition for normalization constant. Thus, Alice's expected payment is bounded by $e_A(N_A^*)+\epsilon+o(1)$. Moreover, when Bob chooses $N_B^*$, to obtain a strictly positive utility in the limit, Alice must play $N_A$ such that $N_A U_G N_B^*\succeq U_A^{\star}=N_A^{\star} U_G N_B^*$. Due to the fact that $U_A^{\star}$ is non-degenerate, $N_A \succeq N_A^{\star}$. Thus, Alice needs to spend at least $e_A(N_A^{\star})\geq e_A(N_A^{*})-\epsilon$ effort. This implies that Alice's utility is bounded by $2\epsilon+o(1)$. Moreover, since we pick $U_A^{\star}\prec U^*, U_A^{\star}\notin \partial{\downarrow U^* }$, we have \[ \lim_{\alpha\rightarrow \infty} \vmi^{w^{\bm{\alpha}(U_A^{\star})}}(U^*)=\mathbbm{1}(U^*\succeq U_A^{\star})=1\] whose square will also converge to 1. Therefore, when Bob chooses $N_B^*$, choosing $N_A^{*}$ will give Alice $e_A(N_A^*)+\epsilon-o(1)$ payment and $\epsilon-o(1)$ utility. We have analogous analysis for Bob's side. Thus, given $\delta>\epsilon>0$, for sufficiently large $\alpha>0$, $(N_A^*, N_B^{*})$ is a $\delta$-equilibrium. Combining the above analysis, the requester's utility will be at least $v(U^*)-\epsilon-e_A(N_A^*)-\epsilon-e_B(N_B^*)-\epsilon-o(1)\geq v(U^*)-e_A(N_A^*)-e_B(N_B^*)-3\epsilon-o(1)\geq  v^*-4\epsilon$ for sufficiently large $\alpha$. 

It's left to analyze the requirement for the number of tasks. For even $C$, the degree of the polynomials $(\vmi^{\phi_\alpha})^2,(\vmi^{\psi_\alpha})^2$ is $2(\alpha-C^2+C(C-1))=2(\alpha-C)$, thus we only need at least $2(\alpha-C)$ tasks to implement the above mechanism. For odd $C$, since $\vmi^{\phi_\alpha},\vmi^{\psi_\alpha}$ are already polynomials and we can use them directly such that we only need $\alpha-C$ number of tasks.
\end{proof}

\section{Conclusion and Discussion}\label{sec:conclusion}

We provide a novel construction of a new family of mutual information measures, volume mutual information (VMI). Aiding by VMI, we construct a family of dominantly truthful and practical multi-task peer prediction mechanisms, VMI-Mechanisms. Moreover, we provide a tractable effort incentive optimization goal for multi-task peer prediction. We show that with this goal, the optimal payment scheme is the threshold payment scheme and there always exists a sequence of dominantly truthful and practical multi-task peer prediction mechanisms, VMI-Mechanisms, that are approximately optimal.

Though the construction of approximately optimal VMI-Mechanisms requires us to perfectly know the optimal threshold, we believe this work provides the first step for optimization over dominantly truthful and practical multi-task peer prediction mechanisms. One important future direction is to relax the modeling assumption for optimization. For example, when we do not perfectly know the threshold, we can use proper densities (e.g. a smaller $\alpha$ with more uncertainty) to obtain a more robust mechanism. The approximation gradually increases the requirement for the number of tasks. When given the constraint for the number of tasks, another future direction is to use a computer-aided approach to optimize over VMI-Mechanisms directly. 

Moreover, we provide a visualization that eases the understanding of mutual information measures. Additionally, this visualization naturally leads to a visual way to fully classify all monotone mutual information in the binary case by the shape of contours. We hope this visualization in binary can also provide insights for the non-binary case. 

\newpage

\bibliographystyle{ACM-Reference-Format}
\bibliography{reference}

\appendix

\section{Basic Measure Theory} \label{sec:measure}

This section introduces several basic concepts in measure theory for rigorousness. However, a measure is just a generalization of the concepts of traditional length, area, and volume. Thus, readers can skip this section and still understand the proof in an intuitive way.  

\paragraph{Measure, integral and monotonicity \cite{simon1983lectures}} We first introduce the concept of measure space $(X,\Sigma,\mu)$. Intuitively, this measure space provides a way to measure the \emph{volume} of the set $X$'s subset. We then introduce the concept of integral. Intuitively, when the $X$ has a density, integral allows us to measure the \emph{volume} of the set $X$'s subset with this density. We require the definition of measure and integral to satisfy monotonicity: any set's volume must be greater than its subset's volume.

Let $X$ be a set. $\Sigma$ is a collection of $X$'s subsets that contains $X$ itself and is closed under complement and countable unions. $\mu$ is a non-negative function $\Sigma:\mapsto \mathbb{R}^+$. The members of $\Sigma$ are called \emph{measurable} sets. For every $A\in \Sigma$, $\mu(A)$ can be seen as $A$'s \emph{volume}. We require the $\mu$ here to satisfy \emph{monotonicity}: for every two measurable sets $A_1\subset A_2$, $\mu(A_1)\leq \mu(A_2)$. We call $(X,\Sigma,\mu)$ is a measure space. 

We also need the definition of \emph{integral} such that there exists a class of integrable functions $f$ where $\int_X f d\mu $ is well-defined. We require this integral definition to satisfy monotonicity as well: for two integrable real-valued functions $f\leq g$ on $X$, \[\int_X f d\mu\leq \int_X g d\mu.\] For a non-negative integrable function $w$ on $X$, $\int_A w d\mu$ can be seen as $A$'s volume with density $w$.  

\paragraph{Lebesgue measure and Hausdorff measure \cite{simon1983lectures}} Here we introduce two measures for the Euclidean space and their relationship. Intuitively, to provide a measure for any triangle's area on $\mathbb{R}^2$, the 2-dimensional Lebesgue measure $\mathscr{L}^2$ works. However, $\mathscr{L}^2$ will assign zero to the measure of any curve in $\mathbb{R}^2$. To provide a measure for any curve's length in $\mathbb{R}^2$, we need the Hausdorff measure $\mathscr{H}^1$.

Formally, the Lebesgue measure $\mathscr{L}^n$ is a measure on $\mathbb{R}^n$. The $\mathscr{L}^n$ of the unit cube $[0,1]^n$ is 1. The Hausdorff measure $\mathscr{H}^m,m\leq n$ is a $m$-dimensional measure on $\mathbb{R}^n$. It agrees with the classical mapping area of an embedded manifold, but it is defined for all subsets of $\mathbb{R}^n$. For Euclidean space, we use the Lebesgue measure as the default measure. That is, $\int_{E\subset\mathbb{R}^n} f(x)dx:=\int_{E\subset\mathbb{R}^n} f(x)d\mathscr{L}^n(x)$.

One commonly used technique in integration is \emph{change of variables}, which needs the \emph{area formula}. For example, when we map a square in $\mathbb{R}^2$ into a parallelogram in $\mathbb{R}^3$ via an affine transformation. The area formula shows how to calculate the area of the parallelogram (Figure \ref{fig:area}). 
\begin{figure}[!ht]
\begin{center}
    \includegraphics[width=10cm]{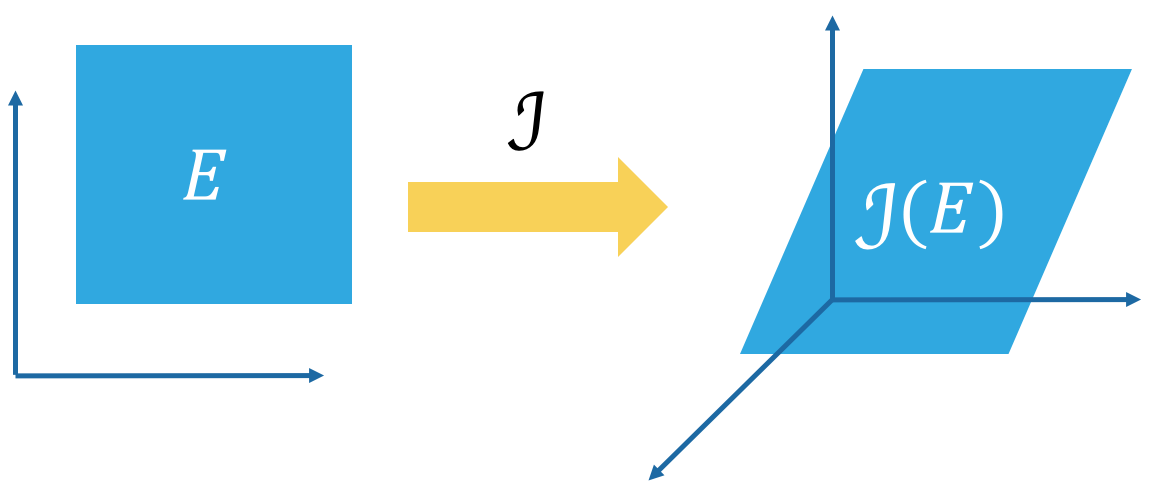}
\end{center}

\caption{An affine transformation $\mathcal{J}$ maps a set $E$ in $\mathbb{R}^2$ to $\mathcal{J}(E)$ in $\mathbb{R}^3$. The area formula will give the ratio of these two sets' areas. We use the Lebesgue measure $\mathscr{L}^2$ to measure the area of $E$ and the Hausdorff measure $\mathscr{H}^2$ to measure the area of $\mathcal{J}(E)$.}\label{fig:area}
    
\end{figure}

\begin{fact}(Area Formula/Change of Variables \cite{simon1983lectures})\label{fact:area}
    Let $\mathcal{J}:\mathbb{R}^m\mapsto \mathbb{R}^n,m\leq n$ be a one to one affine transformation where $\mathcal{J}(\mathbf{v})=\mathbf{J}\mathbf{v}+\mathbf{v}_0$. Let $E\subset \mathbb{R}^m$ be a measurable set and $f:E\mapsto \mathbb{R}^+$ be an integrable function, then
    \[ \int_{\mathcal{J}(E)} f(\mathcal{J}^{-1}(\mathbf{v'}))d\mathscr{H}^m(\mathbf{v'})=\int_E f(\mathbf{v}) \sqrt{\det( \mathbf{J}^{\top}\mathbf{J})} d\mathscr{L}^m(\mathbf{v})=\int_E f(\mathbf{v}) \sqrt{\det( \mathbf{J}^{\top}\mathbf{J})} d\mathbf{v}\]
\end{fact}

\begin{corollary}
    Let $\mathcal{J}:\mathbb{R}^m\mapsto \mathbb{R}^n,m\leq n$ be a one to one affine transformation where $\mathcal{J}(\mathbf{v})=\mathbf{J}\mathbf{v}+\mathbf{v}_0$. Let $E\subset \mathbb{R}^m$ be a measurable set, then by defining $\vol(E):=\mathscr{L}^m(E)$, $\vol(\mathcal{J}(E)):=\mathscr{H}^m(\mathcal{J}(E))$, 
    \[ \vol(\mathcal{J}(E))=\sqrt{\det( \mathbf{J}^{\top}\mathbf{J})} \vol(E)\]
\end{corollary}


\section{Proof of Theorem~\ref{thm:vmi}}

The part 1, $\vmi^w$ is information-monotone and non-negative, follows directly from Lemma~\ref{lem:key} and Lemma~\ref{lem:mono}. To show $\vmi^w$ vanishes on independent variables, notice that the uninformative distributions form a $C-1$-space whose dimension is strictly less than $C(C-1)$. Thus, the $C(C-1)$-Hausdorff measure gives it zero volume. To show the rest of the results, here we introduce two linear algebra operations and their properties that the proof will use.

\subsection{Some Linear Algebra Operations}

\paragraph{Kronecker product} The Kronecker product \cite{vec} of matrix $\mathbf{A}=[A_{ij}]_{ij}\in\mathbb{R}^{m\times n}$ and matrix $\mathbf{B}\in\mathbb{R}^{p\times q}$ is defined as a $mp\times nq$ matrix such that $\mathbf{A}\otimes \mathbf{B}=\begin{bmatrix}
    A_{11} \mathbf{B}&\cdots & A_{1n} \mathbf{B}\\
    \cdots&\cdots & \cdots\\
    A_{n1} \mathbf{B}&\cdots & A_{nn} \mathbf{B}\\
\end{bmatrix}$

\begin{fact} \cite{vec}\label{fact:otimes}
Here are several properties of the Kronecker product. 
\begin{itemize}
    \item Transpose: $(\mathbf{A}\otimes \mathbf{B})^{\top}=\mathbf{A}^{\top}\otimes \mathbf{B}^{\top}$
    \item Determinant: let $\mathbf{A}$ be a $n\times n$ matrix and $\mathbf{B}$ be a $m\times m$ matrix, \[\det((\mathbf{A}\otimes \mathbf{B}))=\det(\mathbf{A})^m \det(\mathbf{B})^n\]
    \item Mixed-product: $(\mathbf{A}\otimes \mathbf{B})(\mathbf{C}\otimes \mathbf{D})=(\mathbf{A}\mathbf{C})\otimes (\mathbf{B}\mathbf{D})$
\end{itemize}

\end{fact}

\paragraph{Vectorization}
For a matrix $\mathbf{A}\in\mathbb{R}^{m\times n}$, the vectorization \cite{vec} of $\mathbf{A}$ is defined as a $m*n$-dimensional column vector $\mvec(\mathbf{A})$ by stacking all column vectors of $\mathbf{A}$ one under the other. For example, when $\mathbf{A}=\begin{bmatrix}
    1 & 3\\ 2 & 4
\end{bmatrix}$, $\mvec(\mathbf{A})=\begin{bmatrix}
    1 \\ 2\\ 3 \\ 4
\end{bmatrix}$. 

\begin{fact} \cite{vec}\label{fact:vec}
    $\mvec(\mathbf{A}\mathbf{B}\mathbf{C})=(\mathbf{C}^{\top}\otimes\mathbf{A})\mvec(\mathbf{B})$
\end{fact}

\paragraph{Proof of the Part 2: VMI provides an interpretation of \dmi: $\dmi^{C-1}\propto \vmi$}

In this part, we show that $\vmi(X;Y)\propto \dmi(X;Y)^{C-1}$.

\begin{proof}

\begin{align*}
    \vmi(X;Y)= \vol(\downarrow  U_{X,Y})=\mathscr{H}^{C(C-1)}(\downarrow  U_{X,Y})
\end{align*}

For simplicity, we replace $U_{X,Y}$ by $U$. It's left to calculate the volume of $\downarrow  U$. We will show that for every $U$, $\mathscr{H}^{C(C-1)}(\downarrow  U)\propto |\det(U)|^{C-1}$. We will first show that $\downarrow  U$ is an affine mapping $\mathcal{J}$ from a subset $E$ in $\mathbb{R}^{C(C-1)}$ whose volume is non-zero and bounded. $E$ is also independent of $U$. With the corollary of the area formula, the volume of $\downarrow  U$ will be proportional to $\sqrt{\det( \mathbf{J}^{\top}\mathbf{J})} $. It's only left to show that $\sqrt{\det( \mathbf{J}^{\top}\mathbf{J})}\propto |\det(U)|^{C-1} $. 


\paragraph{Constructing $E$}

For each column-stochastic matrix $T$, we define $T_*:=T(1:C-1,1:C)$. Let $E$ be the set of all possible $\mvec(T_*)$. Note that $E$ is in $\mathbb{R}^{C(C-1)}$. The following claim shows that the volume of $E$ is non-zero and bounded. 

\begin{claim}
    \[0<\vol(E):=\mathscr{L}^{C(C-1)}(E)<1\]
\end{claim}

\begin{proof}
The set of $T_*$ denotes all $(C-1)\times C$ matrix whose every entry is non-negative and every column sums to a real number in $[0,1]$. Thus, $\vol(E)< 1$ and for a $C(C-1)$ vector, if its every entry is in $[0,\frac{1}{C}]$, then it must be in $E$. Therefore, 

\begin{align*}
    \vol(E):=\mathscr{L}^{C(C-1)}(E)\geq \int_{t_1=0}^{\frac{1}{C}}\int_{t_2=0}^{\frac{1}{C}}\cdots \int_{t_{C(C-1)}=0}^{\frac{1}{C}} d t_1 d t_2 \cdots d t_{C(C-1)}>0
\end{align*}
Thus, $0<\vol(E)<1$. 
\end{proof}

\paragraph{Constructing $\mathcal{J}$}
We start to construct an affine mapping $\mathcal{J}$ from $E$ to $\downarrow  U$. Since every column of $T$ sums to 1, we can represent $\mvec(T)$ as an affine transformation of $\mvec(T_*)$:

\[ \mvec(T)= (\mathbf{I}\otimes \mathbf{W})\mvec(T_*)+\mathbf{c}_0 \]

Here $\mathbf{I}$ is a $C\times C$ identity matrix. $\mathbf{W}$ is a $C\times (C-1)$ matrix, which is a $(C-1)\times (C-1)$ identity matrix with an additional all $-1$ row:

\[ \mathbf{W}:=\begin{bmatrix}
     1 & 0 & \cdots & 0\\
     0 & 1 & \cdots & 0\\
     \cdots & \cdots & \cdots & \cdots\\
     0 & 0 & \cdots & 1\\
    -1 & -1 & \cdots & -1
\end{bmatrix} \]

$\mathbf{c}_0$ is a $C(C-1)$-dimensional column vector where all entries are zero except that the $C^{th}, 2C^{th}, 3C^{th},\cdots$ entries are all one. 

For each element $\mvec(TU)\in \downarrow  U$, 

\begin{align*}
\mvec(TU) & = (U^{\top}\otimes \mathbf{I}) \mvec(T)\\
& = (U^{\top}\otimes \mathbf{I})( (\mathbf{I}\otimes \mathbf{W})\mvec(T_*)+\mathbf{c}_0)\\
& = (U^{\top}\otimes \mathbf{W})\mvec(T_*)+ (U^{\top}\otimes \mathbf{I}) \mathbf{c}_0
\end{align*}

The first equality uses Fact~\ref{fact:vec} and the third equality uses Fact~\ref{fact:otimes}. 

Therefore, $\downarrow  U$ is the image of $E$ with an affine transformation. The corresponding matrix $\mathbf{J}=U^{\top}\otimes \mathbf{W}$.

By applying the area formula (Fact~\ref{fact:area}), 

\begin{align*}
    \vol(\downarrow  U)&= \mathscr{H}^{C(C-1)}(\downarrow  U)\\
    &=\sqrt{\det(\mathbf{J}^{\top}\mathbf{J})}\vol(E)\\ \tag{Fact~\ref{fact:otimes}}
    &= \sqrt{\det((UU^{\top})\otimes (\mathbf{W}^{\top}\mathbf{W}) )}\vol(E)\\ \end{align*}

$\mathbf{W}^{\top}\mathbf{W}$'s dimension is $(C-1)\times (C-1)$ and $UU^{\top}$'s dimension is $C\times C$. Moreover, $\mathbf{W}^{\top}\mathbf{W}=\mathbf{I}+\mathbf{1}$ where $\mathbf{1}$ is a $(C-1)\times (C-1)$ matrix whose entries are all 1. By Gaussian elimination and induction, we can show that the determinant of $\mathbf{W}^{\top}\mathbf{W}$ is $C$. Therefore, based on the determinant property of Kronecker product (Fact~\ref{fact:otimes}), we have 

\begin{align*}
    \vol(\downarrow  U)     &= \sqrt{\det((UU^{\top})\otimes (\mathbf{W}^{\top}\mathbf{W}) )}\vol(E)\\ 
    &= C^{\frac{C}{2}} |\det(U)|^{C-1}  \vol(E)\\ \tag{$0<\vol(E)<1$}
    &\propto |\det(U)|^{C-1}
\end{align*}

\end{proof}

\paragraph{Proof of Part 3: Polynomial Volume Mutual Information}

In this part, we will show that when the density function $w$ is a non-negative degree $d_w$ polynomial, when $C$ is an odd number, $\vmi^D$ is a degree $d_w+C^{C-1}$ polynomial and when $C$ is an even number, $\dmi* \vmi^D$ is a degree $d_w+C^C$ polynomial. Moreover, $\dmi* \vmi^D$ is also information-monotone. 

\begin{proof}

The proof of part 2 shows that $\downarrow  U$ is an affine mapping $\mathcal{J}$ from a subset $E$. Recall that $E$ is the set of all possible $\mvec(T_*)$ where $T_*:=T(1:C-1,1:C)$. We also proved that $\sqrt{\det(\mathbf{J}^{\top}\mathbf{J})}=C^{\frac{C}{2}} |\det(U)|^{C-1}$. Thus, by changing the variables, we have 

\begin{align*}
    \vmi^w(U)&=\int_{x \in \downarrow  U}w(x)d\mathscr{H}^{C(C-1)}(x)\\&=C^{\frac{C}{2}}  |\det(U)|^{C-1} \int_{\forall j, t_{1j}+t_{2j}+\cdots t_{c-1,j}\leq 1,\forall i, t_{ij}\geq 0}w(TU)(\Pi_{j} dt_{1j}dt_{2j}\cdots dt_{c-1,j})\\
\end{align*}

When $w(U)$ is a degree $d_w$ polynomial of entries of $U$, then $w(TU)$ is also a polynomial of the entries of $U$ and $T$. Moreover, fixing $T$, $w(TU)$ is still a degree $d_w$ polynomial for $U$. 

$w(TU)$ can be written as the sum of terms of format $h(T)U_{c_1,c'_1}U_{c_2,c'_2}\cdots U_{c_k,c'_k},k\leq d_w$. We can take out $U_{c_1,c'_1}U_{c_2,c'_2}\cdots U_{c_k,c'_k}$ and only integrate $h(T)$. Thus, after integration, \[\int_{\forall j, t_{1j}+t_{2j}+\cdots t_{c-1,j}\leq 1,\forall i, t_{ij}\geq 0}w(TU)(\Pi_{j} dt_{1j}dt_{2j}\cdots dt_{c-1,j})\] is still a degree $d_w$ formula for $U$'s entries. Note that when $C$ is an odd number, $|\det(U)|^{C-1} = (\det(U))^{C-1} $ is a degree $C(C-1)$ polynomial. When $C$ is even number, we can multiply $\dmi$ to avoid the absolute $|\cdot|$ symbol but still keep the information-monotonicity (the multiplication of two non-negative monotone functions are still monotone). Therefore, when $C$ is an odd number, $\vmi^D$ is a degree $d_w+C(C-1)$ polynomial. When $C$ is an even number, $\dmi* \vmi^D$ is a degree $d_w+C^2$ polynomial and an information-monotone measure.




\end{proof}





\section{Additional proofs}\label{sec:additionalproof}
{
\renewcommand{\thetheorem}{\ref{obs:binary}}

\begin{observation}
In the binary case, there is a one to one mapping from $[0,1]^3$ to $L$. In fact,
\[L=\{\begin{bmatrix} s & t\\ 1-s & 1-t\end{bmatrix} \begin{bmatrix}p & 0\\ 0 & 1-p\end{bmatrix}|s,t,p\in[0,1] \}.\]

Fixing $p$, there is an one to one mapping from $[0,1]^2$ to $\downarrow  U_p$ where $U_p= \begin{bmatrix}p & 0\\ 0 & 1-p\end{bmatrix}$ and $L=\{\downarrow  U_p|p\in[0,1]\}$.

\end{observation}

\addtocounter{theorem}{-1}
}

\begin{proof}

We use $L_1$ to denote $\{\begin{bmatrix} s & t\\ 1-s & 1-t\end{bmatrix} \begin{bmatrix}p & 0\\ 0 & 1-p\end{bmatrix}|s,t,p\in[0,1] \}$. 

It's easy to verify that $\begin{bmatrix} s & t\\ 1-s & 1-t\end{bmatrix} \begin{bmatrix}p & 0\\ 0 & 1-p\end{bmatrix}$ is a joint distribution matrix. Thus, $L_1\subset L$ and we have a natural mapping $U=\begin{bmatrix} s & t\\ 1-s & 1-t\end{bmatrix} \begin{bmatrix}p & 0\\ 0 & 1-p\end{bmatrix}$ from $(s,t,p)\in [0,1]^3$ to $U\in L$. For another direction, for every $U=\begin{bmatrix}u_{00} & u_{01}\\ u_{10} & u_{11}\end{bmatrix}$, we can set $p=u_{00}+u_{10}$ and $s=\frac{u_{00}}{p}$, $t=\frac{u_{01}}{1-p}$ such that \[\begin{bmatrix} s & t\\ 1-s & 1-t\end{bmatrix} \begin{bmatrix}p & 0\\ 0 & 1-p\end{bmatrix}=U.\] Thus, $L\subset L_1$ and there is a mapping from $U=\begin{bmatrix}u_{00} & u_{01}\\ u_{10} & u_{11}\end{bmatrix}\in L$ to $(s=\frac{u_{00}}{u_{00}+u_{10}},t=\frac{u_{01}}{1-(u_{00}+u_{10})},p=u_{00}+u_{10})\in [0,1]^3$. 
\end{proof}

{
\renewcommand{\thetheorem}{\ref{lem:binaryvmi}}

\begin{lemma}
In binary case, 
\begin{align*}
    \vmi^{w}(U)=& 2|\det(U)|\int_{s=0}^{1}\int_{t=0}^1 w(\begin{bmatrix} s & t\\ 1-s & 1-t\end{bmatrix} U )d s d t\\
    =& 2|u_{00}u_{11}-u_{10}u_{01}|\int_{s=0}^{1}\int_{t=0}^1 w(\begin{bmatrix} s & t\\ 1-s & 1-t\end{bmatrix} \begin{bmatrix}u_{00} & u_{01}\\ u_{10} & u_{11}\end{bmatrix}
 )d s d t
\end{align*}
\end{lemma}

\addtocounter{theorem}{-1}
}

\begin{proof}

Let $E$ be $[0,1]^2$ and 
\begin{align*}
    \mathcal{J}(s,t):=&\mvec(\begin{bmatrix} s & t\\ 1-s & 1-t\end{bmatrix} \begin{bmatrix}u_{00} & u_{01}\\ u_{10} & u_{11}\end{bmatrix})\\
    =&\begin{bmatrix}u_{00}&u_{10}\\-u_{00}&-u_{10}\\u_{01}&u_{11}\\-u_{01}&-u_{11}\end{bmatrix}\begin{bmatrix}s\\t\end{bmatrix}+\begin{bmatrix}0\\u_{00}+u_{10}\\0\\u_{01}+u_{11}\end{bmatrix}\\
\end{align*}
    By applying the area formula (Fact~\ref{fact:area}), we have

    \begin{align*}
    \vmi^{w}(U)&= \int_{ x\in \downarrow  U} w(x)d\mathscr{H}^2(x) \\
    &= 2|\det(U)|\int_{s=0}^{1}\int_{t=0}^1 w(\begin{bmatrix} s & t\\ 1-s & 1-t\end{bmatrix} U )d s d t\\
    &= 2|u_{00}u_{11}-u_{10}u_{01}|\int_{s=0}^{1}\int_{t=0}^1 w(\begin{bmatrix} s & t\\ 1-s & 1-t\end{bmatrix} \begin{bmatrix}u_{00} & u_{01}\\ u_{10} & u_{11}\end{bmatrix}
 )d s d t
    \end{align*}
\end{proof}

{
\renewcommand{\thetheorem}{\ref{lem:dirich}}

\begin{lemma}
Given a joint distribution $U^{\star}$, we set $\bm{\alpha}(U^{\star})=\{\alpha u^{\star}_{ij},i,j\in[C],\alpha>0\}$, for all $U$ such that $U^{\star} \notin \partial{\downarrow U}$, 
\[ \lim_{\alpha\rightarrow \infty} \vmi^{w^{\bm{\alpha}(U^{\star})}}(U)=\mathbbm{1}(U\succeq U^{\star})\] where $\partial{\downarrow U}$ is the boundary of $\downarrow U$.
\end{lemma}

\addtocounter{theorem}{-1}
}

\begin{proof}[Proof of Lemma~\ref{lem:dirich}]
We first show that for $U\notin slice(U^{\star})$, $\lim_{\alpha\rightarrow \infty} \vmi^{w^{\bm{\alpha}(U^{\star})}}(U)=0$.
\begin{align*}
\vmi^{w^{\bm{\alpha}(U^{\star})}}(U)\leq & \vol^{w^{\bm{\alpha}(U^{\star})}} (slice(U))\\
=&\int_{x\in slice(U)} w^{\bm{\alpha}(U^{\star})}(x)  d\mathscr{H}^{C(C-1)}(x)\\
=& \frac{1}{C(\bm{\alpha})} \int_{\forall j, \sum_i x_{ij}=\sum_i u_{ij}} \Pi_{ij} x_{ij}^{\alpha_{ij}-1}  d\mathscr{H}^{C(C-1)}(x)\\ \tag{For a single slice, we can integrate independently for each column $j$ and $x_j:=(x_{1j},x_{2j},\cdots,x_{ij})$.}
=& \frac{1}{C(\bm{\alpha})} \Pi_j (\int_{\sum_i x_{ij}=u_j} \Pi_{i} x_{ij}^{\alpha_{ij}-1}  d\mathscr{H}^{C-1}(x_j))\\ \tag{$y_{ij}=\frac{x_{ij}}{u_j}$}
=& \frac{1}{C(\bm{\alpha})} \Pi_j (\int_{\sum_i y_{ij}=1} \Pi_{i} (u_j y_{ij})^{\alpha_{ij}-1} d\mathscr{H}^{C-1}(u_j y_{ij})))\\
\tag{each part is proportional to a Dirichlet density distribution multiplying $u_j^{\alpha_j-1}$ }
= & \Pi_j (\frac{u_j}{u^{\star}_j})^{\alpha_j-1} \\ \tag{$\alpha_j=\alpha u^{\star}_j$}
= & \Pi_j (\frac{u_j}{u^{\star}_j})^{\alpha u^{\star}_j-1}
\end{align*}

Since $\Pi_j (\frac{u_j}{u^{\star}_j})^{u^{\star}_j}<1$, the limit of $\vmi^{w^{\bm{\alpha}(U^{\star})}}(U)$'s upper-bound will be zero as $\alpha$ goes to infinity. 

It's left to analyze the points on $slice(U^{\star})$. Based on the above analysis, we can write the density function on $slice(U^{\star})$ as 
\begin{align*}
\tag{$\sum_i x_{ij}=u^{\star}_j, y_{ij}:=\frac{x_{ij}}{u^{\star}_j}$}
w(x)=\frac{1}{C(\bm{\alpha})} \Pi_j (\Pi_{i} x_{ij}^{\alpha_{ij}-1})=\frac{1}{C(\bm{\alpha})} \Pi_j \left((u^{\star}_j)^{\alpha_{j}-C}(\Pi_{i} y_{ij}^{\alpha_{ij}-1})\right)
\end{align*}

Thus, the normalization constant $C(\bm{\alpha})$ makes the above function a probability density function on $slice(U^{\star})$. For a random $U$ that follows this probability, i.e., $\Pr[U=x]=w(x)$, we will show that it will converge in probability to constant $U^{\star}$. Due to Markov inequality, 

\begin{align*}
\Pr[|U-U^{\star}|^2\geq \epsilon^2]\leq \frac{\sum_{ij}\E[(u_{ij}-u^{\star}_{ij})^2]}{\epsilon^2}
\end{align*}

For all $j$, we use $\mathbf{u}_j$ to represent the $j^{th}$ column vector of matrix $U$. The normalized column vector $\frac{\mathbf{u}_{j}}{u^{\star}_j}$'s density $\Pr[\frac{\mathbf{u}_{j}}{u^{\star}_j}=\mathbf{y}_j]$ is proportional to $\Pi_{i} y_{ij}^{\alpha_{ij}-1}=\Pi_{i} y_{ij}^{\alpha u^{\star}_{ij}-1}$ thus is a Dirichlet-distributed random variable whose expectation is $\frac{\mathbf{u}^{\star}_{j}}{u^{\star}_j}$ and each individual coordinate's variance goes to zero as $\alpha$ goes to infinity (Fact~\ref{fact:beta}).

Thus, $U$ also converges in distribution to constant $U^{\star}$ since convergence in probability implies convergence in distribution \cite{billingsley2013convergence}. In such case, for all continuity set $A$ whose boundary has zero measure in the limit \cite{billingsley2013convergence}, we have $\Pr[U\in A]$ converges to $\Pr[U^{\star}\in A]=\mathbbm{1}(U^{\star}\in A)$. 

\begin{align*}
\vmi^{w^{\bm{\alpha}(U^{\star})}}(U) = \int_{x\in \downarrow U} w^{\bm{\alpha}(U^{\star})}(x) d\mathscr{H}^{C(C-1)}(x)=\Pr[U^{\star}\in \downarrow U]
\end{align*}

Therefore, as long as the boundary of $\downarrow U$ does not contain $U^{\star}$, the above formula will converge to $\mathbbm{1}(U^{\star}\in \downarrow U)=\mathbbm{1}(U\succeq U^{\star})$.

\end{proof}

\section{Divergence families}\label{sec:divergence}

We use $\Sigma$ to denote a discrete set of signals. 
\paragraph{$f$-divergence~\cite{ali1966general,csiszar2004information}} 
$f$-divergence $\mathrm{D}_f:\Delta_{\Sigma}\times \Delta_{\Sigma}\rightarrow \mathbb{R}$ is a non-symmetric measure of the difference between distribution $\mathbf{p}\in \Delta_{\Sigma} $ and distribution $\mathbf{q}\in \Delta_{\Sigma} $ 
and is defined to be $$\mathrm{D}_f(\mathbf{p},\mathbf{q})=\sum_{\sigma\in \Sigma}
\mathbf{p}(\sigma)f\left( \frac{\mathbf{q}(\sigma)}{\mathbf{p}(\sigma)}\right)$$
where $f(\cdot)$ is a convex function and $f(1)=0$. $f$-divergence is non-negative and equals zero if $\mathbf{p}=\mathbf{q}$.

Now we introduce two $f$-divergences in common use: KL divergence, and Total variation Distance.
\begin{example}[KL divergence]
Choosing $-\log(x)$ as the convex function $f(x)$, $f$-divergence becomes KL divergence $D_{KL}(\mathbf{p},\mathbf{q})=\sum_{\sigma}\mathbf{p}(\sigma)\log\frac{\mathbf{p}(\sigma)}{\mathbf{q}(\sigma)}$
\end{example}

\begin{example}[Total Variation Distance]
Choosing $|x-1|$ as the convex function $f(x)$, $f$-divergence becomes Total Variation Distance $D_{tvd}(\mathbf{p},\mathbf{q})=\sum_{\sigma}|\mathbf{p}(\sigma)-\mathbf{q}(\sigma)|$
\end{example}

\paragraph{Proper scoring rules~\cite{winkler1969scoring}}
A scoring rule $PS:  \Sigma \times \Delta_{\Sigma} \rightarrow \mathbb{R}$ takes in a signal $\sigma \in \Sigma$  and a distribution over signals $\mathbf{p} \in \Delta_{\Sigma}$ and outputs a real number.  A scoring rule is \emph{proper} if, whenever the first input is drawn from a distribution $\mathbf{p}$, then $\mathbf{p}$ will maximize the expectation of $PS$ over all possible inputs in $\Delta_{\Sigma}$ to the second coordinate. A scoring rule is called \emph{strictly proper} if this maximum is unique. We will assume throughout that the scoring rules we use are strictly proper. Slightly abusing notation, we can extend a scoring rule to be $PS:  \Delta_{\Sigma} \times \Delta_{\Sigma} \rightarrow \mathbb{R}$  by simply taking $PS(\mathbf{p}, \mathbf{q}) = \E_{\sigma \leftarrow \mathbf{p}}(\sigma,  \mathbf{q})$.  We note that this means that any proper scoring rule is linear in the first term. 

\begin{example}[Log Scoring Rule~\cite{winkler1969scoring,gneiting2007strictly}]
Fix an outcome space $\Sigma$ for a signal $\sigma$.  Let $\mathbf{q} \in \Delta_{\Sigma}$ be a reported distribution.
The Logarithmic Scoring Rule maps a signal and reported distribution to a payoff as follows:
\[LSR(\sigma,\mathbf{q})=\log (\mathbf{q}(\sigma)).\]
\end{example}

\begin{example}[Quadratic Scoring Rule~\cite{winkler1969scoring,gneiting2007strictly}]
Fix an outcome space $\Sigma$ for a signal $\sigma$.  Let $\mathbf{q} \in \Delta_{\Sigma}$ be a reported distribution.
The Quadratic Scoring Rule maps a signal and reported distribution to a payoff as follows:
\[QSR(\sigma,\mathbf{q})= \sum_{\sigma'}(\mathbf{1}_{\sigma}-\mathbf{q}(\sigma'))^2\] where $\mathbf{1}_{\sigma}$ is a $|\Sigma|$-dimensional vector such that $\mathbf{1}_{\sigma}(\sigma)=1$ and $\forall \sigma'\neq \sigma, \mathbf{1}_{\sigma}(\sigma')=0$.  \end{example}

\paragraph{Bregman Divergence~\cite{bregman1967relaxation}} Bregman divergence $\mathrm{D}_{PS}:\Delta_{\Sigma}\times \Delta_{\Sigma}\rightarrow \mathbb{R}$ is a non-symmetric measure of the difference between distribution $\mathbf{p}\in \Delta_{\Sigma} $ and distribution $\mathbf{q}\in \Delta_{\Sigma} $
and is defined to be \[\mathrm{D}_{PS}(\mathbf{p},\mathbf{q})=PS(\mathbf{p},\mathbf{p})-PS(\mathbf{p},\mathbf{q})\]
where $PS$ is a proper scoring rule. Like $f$-divergence, Bregman-divergence is non-negative and equals zero if $\mathbf{p}=\mathbf{q}$. For special Bregman-divergences, the log scoring rule leads to the KL-divergence as well. 

\begin{example}

The quadratic scoring rule leads to quadratic divergence \[D_{QSR}(\mathbf{p},\mathbf{q})=\sum_{\sigma}(\mathbf{p}(\sigma)-\mathbf{q}(\sigma))^2.\] 

\end{example}

\end{document}